\documentclass[11pt]{article}
\usepackage[utf8]{inputenc}
\usepackage[utf8]{inputenc}
\usepackage[english]{babel}
\usepackage{amsfonts,amsmath,amssymb, mathtools}
\usepackage{amsthm, thmtools, thm-restate}
\usepackage{enumerate, graphicx}
\usepackage{xspace}
\usepackage{url}
\usepackage{enumerate, paralist}
 \usepackage{bbm} 
 \usepackage{float} 
\usepackage{multirow} 
\usepackage{algorithm} 
\usepackage[bb=boondox]{mathalfa} 
\usepackage{soul} 
\usepackage{todonotes}
\usepackage{diagbox}
\usepackage[noend]{algpseudocode} 
\usepackage{hyperref} 
\usepackage[margin=1in]{geometry}
\usepackage{setspace}
\usepackage{boxedminipage}
\usepackage{alltt}

\hypersetup{
    colorlinks,
    linkcolor={red},
    citecolor={blue},
    urlcolor={teal}
}

\newtheorem{theorem}{Theorem}
\newtheorem{lemma}{Lemma}[section]
\newtheorem{claim}[lemma]{Claim}
\newtheorem{proposition}[lemma]{Proposition}
\newtheorem{fact}[lemma]{Fact}

\newtheorem{definition}[lemma]{Definition}

\theoremstyle{definition}

\theoremstyle{plain}

\def\floor#1{\lfloor {#1} \rfloor}
\def\ceil#1{\left\lceil {#1} \right\rceil}

\def\ind{\mathbbm{1}}

\newcommand\Var[2][]{
\ensuremath{\mathrm{\mathbf{Var}}_{#1} 
\pmb{\left[\vphantom{#2}\right.}
{#2}
\pmb{\left.\vphantom{#2}\right]}
}}

\newcommand\E[2][]{
\ensuremath{\mathrm{\mathbf{E}}_{#1} 
\pmb{\left[\vphantom{#2}\right.}
{#2}
\pmb{\left.\vphantom{#2}\right]}
}}

\renewcommand\Pr[2][]{
\ensuremath{\mathrm{\mathbf{Pr}}_{#1} 
\pmb{\left[\vphantom{#2}\right.}
{#2}
\pmb{\left.\vphantom{#2}\right]}
}}

\def\DD{\mathcal{D}}

\def\R{\mathbb{R}}

\def\bin{\ensuremath{\mathrm{\mathbf{Bin}}}}

\def\poly{\ensuremath{\mathrm{poly}}}
\def\polylog{\ensuremath{\mathrm{polylog}}}

\newcommand{\amnote}[1]{\footnote{\textcolor{green}{\textbf{Andrew}}: #1}}

\usepackage{graphicx}
\usepackage{amssymb}
\usepackage{framed}
\usepackage{times}

\usepackage{url}
\usepackage{amsmath}
\usepackage{epstopdf}

\newcommand{\bX}{\boldsymbol{X}}
\newcommand{\bZ}{\boldsymbol{Z}}
\newcommand{\bY}{\boldsymbol{Y}}
\newcommand{\bB}{\boldsymbol{B}}
\newcommand{\bi}{\boldsymbol{i}}
\newcommand{\bG}{\boldsymbol{G}}
\newcommand{\bN}{\boldsymbol{N}}
\newcommand{\N}{\mathbb{N}}

\newcommand{\boldeta}{\boldsymbol{\eta}}
\newcommand{\SimpleEstimator}{\texttt{LogEstimator}}
\newcommand{\eps}{\epsilon}
\newcommand{\eqdef}{\mathop{=}^{\tiny\text{def}}}
\newcommand{\Ber}{\mathrm{Ber}}
\newcommand{\NB}{\mathrm{NB}}

\newcommand{\lsim}{\lesssim}

\newcommand{\bc}{\boldsymbol{c}}

\newcommand{\bs}{\boldsymbol{s}}

\def\citep{\cite}

\def\NB{\text{NB}}

\def\Xmax{X_{\max}}
\def\hq{\pmb{\hat q}}
\def\hH{\pmb{\hat H}}
\def\tilH{\tilde H}

\title{Estimation of Entropy in Constant Space \\ with Improved Sample Complexity}

\author{
Maryam Aliakbarpour\thanks{Boston University and Northeastern University, \texttt{maryam.aliakbarpour@gmail.com}.}
\and
Andrew McGregor \thanks{University of Massachusetts Amherst, \texttt{mcgregor@cs.umass.edu}. Supported by NSF awards CCF-1934846, CCF-1908849, and CCF-1637536.}
\and
Jelani Nelson
\thanks{UC Berkeley, \texttt{minilek@berkeley.edu}. Supported by NSF award CCF-1951384, ONR grant N00014-18-1-2562, ONR DORECG award N00014-17-1-2127, and a Google Faculty Research Award.}
\and
Erik Waingarten \thanks{Stanford University, \texttt{eaw@cs.columbia.edu}. Part of this work is supported by the National Science Foundation under Award no. 2002201 and Moses Charikar's Simons Investigator Award.}
}
\date{\today}

\begin{document}

\maketitle

\begin{abstract}
  Recent work of Acharya et al.~(NeurIPS 2019) showed how to estimate the entropy of a distribution $\mathcal D$ over an alphabet of size $k$ up to $\pm\epsilon$ additive error by streaming over $(k/\epsilon^3) \cdot \polylog(1/\epsilon)$ i.i.d.\ samples and using only $O(1)$ words of memory. In this work, we give a new constant memory scheme that reduces the sample complexity to $(k/\epsilon^2)\cdot \polylog(1/\epsilon)$. We conjecture that this is optimal up to $\polylog(1/\epsilon)$ factors.
\end{abstract}

\section{Introduction}\label{sec:intro}
In the field of {\it streaming algorithms}, an algorithm makes one pass (or few passes) over a database while using memory sublinear in the data it sees to then answer queries along the way or at the data stream's end. Researchers have developed various algorithms, as well as memory lower bounds, for such problems for over four decades \citep{MunroP80,MisraG82,AlonMS99}. For the vast majority of research in the field, the database is assumed to be fixed, and algorithms are then analyzed through the lens of worst case analysis.

In this work, we look to further develop the relationship between streaming algorithms and statistics, specifically studying statistical inference through low-memory streaming algorithms. In this setup, rather than processing a worst-case instance of some fixed database, our input is instead a {\it distribution} $\mathcal D$, and our algorithm processes i.i.d.\ samples from $\mathcal D$ with the goal of inferring its properties. Natural questions then arise, such as understanding the tradeoffs between sample complexity, memory, accuracy, and confidence, or even understanding whether a low-memory algorithm exists at all for a particular inference problem even if we allow the streaming algorithm to draw an unlimited number of samples. Work on streaming algorithms for statistical inference problems began in \cite{GuhaM07}, which studied nonparameteric distribution learning, followed by the work of \cite{ChienLM10}, studying low-memory streaming algorithms for use in robust statistics and distribution property testing. Interest in the area later exploded off after work of \cite{SteinhardtVW16}, which explicitly raised the question of whether low memory might place fundamental limits on learning rates, with a flurry of works proving such limitations in response \citep{MoshkovitzM17,KolRT17,Raz17,GargRT18,SharanSV19,GargRT19,GargKR20,GargKLR21}, starting with a work of \cite{Raz19} on memory/sample tradeoff lower bounds for learning parities ($\mathcal D$ generates $(x,\langle w,x\rangle)$ for $x$ uniform in the hypercube with $w$ an unknown parameter, and the goal is to learn $w$).

In this work, following \cite{AcharyaBIS19}, we focus specifically on the problem of estimating the entropy of an unknown distribution $\mathcal D$ over $\{1,\ldots,k\}$ using a low-memory streaming algorithm over i.i.d.\ samples. It is known that to estimate the entropy up to $\epsilon$ additive error with large constant success probability, without memory constraints the optimal sample complexity is \[n = \Theta \left (\max \left \{\frac 1{\epsilon}\frac k{\log(k/\epsilon)}, \frac{\log^2 k}{\epsilon^2}\right \}\right )\] \citep{ValiantV17,ValiantV11,JiaoVHW15,WuY16}. Prior work by \cite{BatuDKR05} also shows that a sublinear number of samples is possible for multiplicative approximation of entropy for distributions whose entropy is sufficiently large.
The known optimal algorithms from prior work, however, must remember all samples and hence use $\Omega(n)$ words of memory\footnote{As in prior work, we use a ``word'', or ``machine word'', to denote a unit of memory that can hold $\Theta(\log(k/\epsilon))$ bits. Essentially, a machine word is large enough to hold the name of an item in the alphabet, as well as the value of $\epsilon$.}. The algorithm of \cite{AcharyaBIS19} uses only $O(1)$ words of memory, though at the cost of requiring an increased sample complexity of $k \cdot \tilde O( 1/{\epsilon^3})$\footnote{We use $\tilde O(f)$ to denote a function which is $O(f\cdot \textup{poly}(\log f))$}. In this work, our goal is to address the question: {\it to what extent was the worsening of sample complexity in previous work necessary to achieve constant memory?}

\paragraph{Our Contribution.} We show that using $O(1)$ words of memory\footnote{More precisely, we provide a uniform algorithm which given any $k,\epsilon$ generates a program with source code of size $O(\log\log(1/\epsilon))$ words, and that fixed program can then process any stream in $O(1)$ words of working memory; see Section~\ref{sec:bit-complexity} for details.}, it is possible to obtain a sample complexity of $k\cdot\tilde O(1/{\epsilon^2})$, which is an improvement over the previous memory-efficient sample complexity bound which had cubic dependence on $1/\epsilon$. The starting point of our algorithm revisits a simple estimator proposed by \cite{AcharyaBIS19}. Their simple estimator uses $O(k \log^2(k/\eps)/\eps^3)$ samples to estimate the entropy in constant space. Our novel contribution is a modification which estimates a bias incurred by the estimator; this change allows us to use only $O(k \log^2 k \log^2(1/\eps) / \eps^2)$ samples. 
With the simple estimator with improved sampled complexity in hand, we show how an ``interval-based'' algorithm, similar to the one in \cite{AcharyaBIS19}, improves the dependence on $k$ to $k \cdot \tilde{O}(1/\eps^2)$.

We remark that there has been other work on estimating entropy in the data streaming model \citep{BhuvanagiriG06,ChakrabartiCM10,HarveyNO08}, but those works are qualitatively different from our own current work and that of \cite{AcharyaBIS19}. Specifically, they take the worst case point of view, where the stream items are  not drawn i.i.d.\ from a distribution, but rather the stream itself is viewed as a worst-case input and the goal is to estimate its empirical entropy. In that model, $O(1)$ memory algorithms for $\pm \epsilon$ additive estimation to entropy provably do not exist, as there is a known memory lower bound of $\Omega(1/(\log^2(1/\epsilon)\epsilon^2))$ \citep{ChakrabartiCM10}.

\paragraph{Overview of Approach.} We start by describing the basic algorithm of \cite{AcharyaBIS19}. Their basic estimator takes a single random sample $\bi\sim \mathcal D$, followed by $N$ more i.i.d.\ samples. Then, they define $\bN_x$ to be the number of these $N$ samples equal to $\bi$. The estimate $\hat p_{\bi} := \bN_x / N$ is an unbiased estimator of of the probability $p_{\bi}$ of $\bi$ according to $\mathcal D$, and for large $N$, $\log(1/\hat p_{\bi})$ is a reasonable estimator for the entropy $H = H(\mathcal D) = \E{\log(1/p_{\bi})}$ of $\mathcal D$. One can then average many such independent estimates. There is an additional technical detail, that $\hat p_{\bi}$ may be zero (if $\bN_x$ is zero), which is fixed via a ``one-smoothing'' trick of actually setting $\hat p_{\bi} := (\bN_x + 1) / N$ (which introduces an acceptably small amount of additional bias when $N$ is sufficiently large).

Our improvement begins with the observation that $\log(1/\hat p_{\bi})$ is {\it not} an unbiased estimator for $H$. We first propose a similar but different estimator to the previous simple estimator. We also begin by taking a random sample $\bi \sim \mathcal{D}$; however, rather than letting $\bN_x$ be sampled from the binomial distribution $\mathsf{Bin}(N, p_{\bi})$, we sample a negative binomial random variable $\bX$, which is the number of additional draws to see $i$ exactly $t$ more times ($t$ is a parameter of the algorithm). Henceforth we let $\NB(t, p)$ denote such a negative binomial random variable, where the underlying Bernoulli experiment has success probability $p$. Then $\E{\bX} = t/p_{\bi}$, and we will use $\log(\bX/t)$ as a reasonable estimate of $\log(1/p_{\bi})$. This estimator is also biased, but we can correct for this bias using a few more samples. 

Specifically, let $\bY = \bX p_{\bi}/t$ and consider the degree-$r$ Taylor expansion of our estimate $\log(\bX/t)$ and the ideal quantity $\log(1/p_{\bi})$. As it will turn out, the expectation of the degree-$r$ Taylor expansion of $\log(\bX/t) - \log(1/p_{\bi}) = \log \bY$ is a degree-$r$ polynomial in $p_{\bi}$. By drawing $r$ additional samples, we may design an estimator for this polynomial, and subtract it from $\log(\bX / t)$. Correcting some of the bias in this way gives us our improved estimate for $\log(1/p_{\bi})$. Our analysis of this scheme shows that a sample complexity of $(k/\epsilon^2) \cdot \polylog(k/\epsilon)$ suffices. We then describe and analyze an improved algorithm in Section~\ref{sec:bucketing}, which achieves $(k/\epsilon^2) \cdot \polylog(1/\epsilon)$ sample complexity by additionally incorporating a ``bucketing'' scheme, similar to one proposed in \cite{AcharyaBIS19}. The idea is to partition the possibilities for values of $\bX$ into disjoint intervals $I_\ell = [b_{\ell-1},b_\ell)$ for $\ell=1,2,\ldots,L$ and optimized choices of the breakpoints $b_\ell$, then estimate both $\Pr{\bX\in I_\ell}$ and the conditional contributions to entropy conditioned on $\bX\in I_\ell$ for each $\ell$. By estimating separately for each $I_\ell$, one can show that the conditional variance is reduced to obtain an overall smaller sample complexity of $(k/\epsilon^2)\cdot \polylog(1/\epsilon)$, a strict improvement over that of \cite{AcharyaBIS19}; details are in Section~\ref{sec:bucketing}.

\section{A Simple Algorithm and Analysis}\label{sec:simple-alg}

Let $k \in \N$, and $\mathcal{D}$ be an unknown distribution supported on $[k]$. For any $i \in [k]$, we denote the probability that $i \in [k]$ is sampled by $\mathcal{D}$ as $p_i$. The goal is to design a low-space streaming algorithm which receives independent samples from $\mathcal{D}$ and outputs an estimate to the entropy:
\[ H(\mathcal{D}) \eqdef \sum_{i=1}^k p_i \log\left( \frac{1}{p_i} \right) = \E[\bi \sim \mathcal{D}]{\log\left( \dfrac{1}{p_{\bi}} \right)}
,\] 
where logarithms above and throughout this paper are base-$2$, unless otherwise stated.

\subsection{An Estimator for $\pmb{\log(1/p_{i})}$}
As mentioned in Section~\ref{sec:intro}, similarly to \cite{AcharyaBIS19} the algorithm aims to estimate $H(\mathcal{D})$ by taking a sample $\bi \sim \mathcal{D}$ and estimating $\log(1/p_{\bi})$. Then, averaging these estimates will give an estimator for $H(\mathcal{D})$ (albeit with a super-linear dependence on $k$, which we fix in Section~\ref{sec:bucketing}). We describe the estimator in Figure~\ref{fig:simple-estimator}. 

\begin{figure}[t!]
\begin{framed}
\noindent Subroutine \SimpleEstimator$(\mathcal{D}, i)$
\begin{flushleft}\noindent {\bf Input:} Sample access to a distribution $\mathcal{D}$ supported on $[k]$, an index $i \in [k]$ where $p_i \neq 0$.

\noindent {\bf Output:}  A number $\boldeta \in \R_{\geq0}$, which is our bias estimate.

\begin{enumerate}

\item\label{en:line-1} We draw enough samples from $\mathcal{D}$ so that $i$ is sampled exactly $t$ times, and let $\bX \in \N$ denote the number of samples taken. 
\item\label{en:line-2} For $r \in \N$, let $f \colon \R \to \R$ denote the degree-$r$ Taylor expansion of $\log z$ centered at $1$, and $h_t \colon [0, 1] \to \R$ be the degree-$r$ polynomial satisfying
\[ h_t(\rho) = 
\E[\bZ \sim \NB(t, \rho)]{f\left( \dfrac{\bZ \cdot \rho}{t} \right)}
.\]
Finally, $g \colon [0,1]^{r} \to \R$ is the linear function with $g(\rho, \rho^2, \dots, \rho^{r}) = h_t(\rho)$. We take $r$ additional independent samples from $\mathcal{D}$, and for $j \in [r]$, we let $\bB_j$ be the indicator random variable that the first $j$ samples were all $i$. Note that $\{\bB_j\}_{j\in [r]}$ can be encoded using a single counter requiring $\log r$ bits.
\item\label{en:line-3} We return 
\[ \boldeta \eqdef \log\left(\frac{\bX}{t}\right) - g\left( \bB_1,  \bB_2, \dots,  \bB_{r}\right). \]
\end{enumerate}
\end{flushleft}
\end{framed}
\caption{Description of the  estimator for $\log(1/p_i)$.}\label{fig:simple-estimator}\label{fig:LogEstimatorDi}
\end{figure}

There are three main steps in the analysis. In the first, we show that the estimator has small bias. The second is showing that the above estimator has low variance. Finally, we show that the estimator may be computed with few bits. In Figure~\ref{fig:simple-estimator}, we set $r = \Theta(\log(1/\eps))$ and $t = \Theta(\log^2(1/\eps))$ to obtain an estimator whose bias is at most $\eps$ and variance is at most $O(\log^2 k)$. It then follows that repeating the estimate of $\log(1/p_{\bi})$ for $O(\log^2 k / \epsilon^2)$ i.i.d.\ chosen $\bi \sim \mathcal{D}$ gives the desired estimate with probability at least $2/3$. These parameter settings establish the following theorem:

\begin{theorem}\label{thm:main-thm-simple}
There exists a single-pass data stream algorithm using $O(1)$ words of working memory that processes a stream of $O(k\epsilon^{-2} \log^2 k \log^2 (1/\epsilon))$ i.i.d.~samples from an unknown distribution $\mathcal{D}$ on $[k]$ and returns an additive $\epsilon$ approximation of $H(\mathcal{D})$ with probability $2/3$.
\end{theorem}

The space complexity in the theorem above follows since computing the estimator just requires maintaining integers in the sets $[k], [t],$ and $[r]$, as well as computing a low-degree polynomial (whose coefficients we bound in Appendix~\ref{sec:bit-complexity}). To compute the average of multiple estimators in small space it suffices to compute the sum of the estimates where each estimator is computed in series. The sample complexity bound (given the specified parameters) in the above theorem follows directly from the sample complexity of \SimpleEstimator. By virtue of the fact our estimators are based on negative binomial distributions ($\bX$ in Figure \ref{fig:LogEstimatorDi} is the number of Bernoulli trials until $t$ successes), this in turn follows directly from the expectation of negative binomial distributions:

\begin{fact}[Expected Sample Complexity of \SimpleEstimator]\label{fact:sample-complexity}
Suppose we draw $\bi \sim \mathcal{D}$ and execute \emph{\SimpleEstimator}$(\mathcal{D}, \bi)$. Then, the expected sample complexity is 
\[ \sum_{i=1}^k p_i \left(r + \frac{t}{p_i} \right) = r + t k.\]
\end{fact}

Although the number of samples we draw is a random variable that is only bounded in expectation, note that it implies the existence of a good algorithm that always has a bounded sample complexity: namely, we can simply terminate the algorithm early and output \textsf{Fail} if it draws a large constant factor times more samples than we expect, which happens with low probability by Markov's inequality.

Before moving on to the showing the properties of the estimator, we verify that $h_t(\rho)$ is a degree-$r$ polynomial.
\begin{lemma}\label{sec:h-poly}
For any $r \in \N$, let $f \colon \R \to \R$ denote the degree-$r$ Taylor expansion of $\log(z)$ centered at $1$. Then, for any $\rho > 0$ and $t \in \N$, 
\[ h_t(\rho) = \E[\bZ \sim \NB(t, \rho)]{f\left( \frac{\bZ \cdot \rho}{t}\right)}\]
is a polynomial of degree at most $r$.
\end{lemma}

\newcommand{\Geo}{\mathrm{Geo}}
\newcommand{\Li}{\mathrm{Li}}

\begin{proof}
Recall that the random variable $\bZ \sim \NB(t, \rho)$ is the number of independent trials from a $\Ber(\rho)$ distribution before one sees $t$ successes. Furthermore, $f$ is the degree-$r$ Taylor expansion of $\log z$ centered at $1$, and 
\[ f(z) = \sum_{i=1}^r \dfrac{(-1)^{i+1}}{i} \cdot (z-1)^r. \]
By linearity of expectation, it suffices to show that for every $j \in \{ 1, \dots, r\}$, $\E[\bZ]{ (\bZ\rho / t - 1)^j}$ is a degree-$j$ polynomial in $\rho$. Note that $\bZ$ is a sum of $t$ independent $\Geo(\rho)$ random variables, so by expanding $( \frac{1}{t}\sum_{i=1}^t \bG_i \rho - 1)^j$ and applying linearity of expectation once more, it suffices to show that
\[ \E[\bG \sim \Geo(\rho)]{ \left(\bG \cdot \rho\right)^j} = \rho^j 
\E[\bG \sim \Geo(\rho)]{ \bG^j } = \rho^j \sum_{k=1}^{\infty} \rho (1-\rho)^{k-1} k^j\]
is a degree-$j$ polynomial in $\rho$. We note that this latter term, $\E[\bG]{\bG^j}$ may be expressed as $\rho \cdot \Li_{-j}(1-\rho)$, where $\Li_{-j}(\cdot)$ is the polylogarithm function (see~\cite{wikipolylogarithm}). $\Li_{-j}(1-\rho)$ happens to be a rational function, where the denominator is exactly $\rho^{j+1}$, which cancels the $\rho^{j+1}$ term   . In addition, the numerator of $\Li_{-j}(1-\rho)$ is a degree-$j$ polynomial in $\rho$, which gives the desired polynomial representation.
\end{proof}

Finally, it will be useful for the variance calculation to show that the correction term is always bounded, which we show here.

\begin{lemma}\label{lem:bound-on-g}
There exists a universal constant $c > 0$ such that, for any $r, t \in \N$, if we let $g \colon [0,1]^r \to \R$ be the linear function where $g(\rho, \rho^2,\dots, \rho^r) = h_t(\rho)$, then $g(b) \in [-c, c]$ for all $b \in \{0,1\}^r$.
\end{lemma}

\begin{proof} 
Recall $g \colon [0,1]^r \to \R$ is the linear function where $g(\rho, \rho^2, \dots, \rho^r) = h_t(\rho)$. Hence, in order to show that $g \colon \{0,1\}^r \to \R$ is bounded, it suffices to show that the sum-of-magnitudes of the $r+1$ coefficients of $h_t$ is bounded. Since we have
\begin{align*}
h_t(\rho) &= \E[\bZ \sim \NB(t, \rho)]{ f\left( \frac{\bZ \cdot \rho}{t}\right)} = \sum_{i=1}^{r} \dfrac{(-1)^{i+1}}{i} \cdot \E[\bZ \sim \NB(t, \rho)]{ \left(\frac{\bZ \cdot \rho}{t} - 1 \right)^i }
\,.
\end{align*}
Notice that in Lemma~\ref{sec:h-poly}, we showed that each $\E[\bZ]{ (\bZ \rho / t - 1)^i}$ is a degree-$i$ polynomial in $\rho$, and the bound (\ref{eq:moment-bound}) implies that, for each $i \in \{1,\dots, r\}$ these polynomials are at most $\left(O( i / \sqrt{t} ) \right)^i$ in magnitude. Furthermore, since these are degree-$i$ polynomials bounded in $[0,1]$, we conclude (by Lemma 4.1 in \cite{Sherstov13}), that the coefficients in $\E[\bZ]{ \left( \bZ \rho / t - 1\right)^{i}}$ are at most $\left( O(i / \sqrt{t}) \right)^i$. 
In particular, we have that the $r$ coefficients of $h_t(\rho)$ are at most 
\[ \sum_{i=1}^r \frac{1}{i} \cdot \left( O(i / \sqrt{t}) \right)^{i}
\leq \sum_{i=1}^r \left( O(i / \sqrt{t}) \right)^{i}
= O(1 / \sqrt{t})\, \]
because $r \ll \sqrt t$.
To show that $g \colon \{0,1\}^r \to \R$ is bounded, we add the magnitudes of the $r$ coefficients, which is $O(r / \sqrt{t}) = O(1)$ when $r = O(\log(1/\eps))$ and $t = O(\log^2(1/\eps))$.
\end{proof}

\subsection{Bounding Bias of Estimator} 
\begin{lemma}\label{lem:main-lemma}
Let $\mathcal{D}$ be any distribution and consider any $i \in [k]$. If, for $\eps \in (0, 1)$, we instantiate $\emph{\SimpleEstimator}(\mathcal{D}, i)$ with $r = \Theta(\log(1/\eps))$ and $t = \Theta(\log^2(1/\eps))$, which produces the random variable $\boldeta$, then
\begin{align*}
\left|\E{\boldeta} - \log\left(\frac{1}{p_i}\right) \right| \leq \eps.
\end{align*}
\end{lemma}

The remainder of the section constitutes the proof of Lemma~\ref{lem:main-lemma}, which will follow from a sequence of claims.

\begin{claim}\label{cl:exp-h-y}
In an execution of $\emph{\SimpleEstimator}(\mathcal{D}, i)$, let $\bX$ and $\boldeta$ be defined as in Line~\ref{en:line-1} and Line~\ref{en:line-3} of Figure~\ref{fig:simple-estimator}, and let $\bY = \bX \cdot p_{i} / t$. Then,
\[ \E{\boldeta} - \log\left(\frac{1}{p_i}\right) = \E[\bX]{h(\bY)},\]
where $h(z)$ is the error in the degree-$r$ Taylor expansion of $\log z$ at $1$.
\end{claim}

\begin{proof}[Proof of Claim~\ref{cl:exp-h-y}]
Notice that $\bX$ is the number of trials from $\Ber(p_i)$ until we see $t$ successes. We now have the following string of equalities:
\begin{align*}
\E[\bX, \bB_1,\dots, \bB_{r}]{ \boldeta - \log\left(\frac{1}{p_i}\right)}
&= 
\E[\bX]{ \log \bY } - 
\E[\bB_1,\dots, \bB_{r}]{ g\left( \bB_1,  \bB_2, \dots, \bB_{r}\right)}
\\
		\E[\bX]{f(\bY) + h(\bY)}
		- g( p_i, p_i^2,\dots, p_i^{r} )= 
		\E[\bX]{h(\bY)}
		\,,
\end{align*}
where we used the fact that $g$ is a linear function, and that $\E{\bB_{\ell}} = p^{\ell}_i$ in order to substitute
\[ \E[\bB_1,\dots, \bB_r]{g(\bB_1,\dots,\bB_r)} = g(p_i, p_i^2,\dots, p_i^r). \]
Furthermore, we divide $\log \bY = f(\bY) + h(\bY)$, where $f(z)$ is the degree-$r$ Taylor expansion of $\log z$ at $1$, and  $h(z) = \log z - f(z)$ is the error in the degree-$r$ Taylor expansion of $\log(z)$, i.e.,
\[ h(z) = \log(z) - f(z) = \sum_{\ell=r+1}^{\infty}(-1)^{\ell+1} \cdot \dfrac{(z-1)^{\ell}}{\ell}. \]
Finally, by construction of $g$, $\E{f(\bY)} = g(p_i,p_i^2, \dots, p_i^r)$, which gives the desired equality.
\end{proof}

\begin{lemma}\label{cl:h-small}
For any $\eps \in (0, 1)$, letting $r = \Theta(\log(1/\eps)$ and $t = \Theta(\log^2(1/\eps))$, we have that for $p_i > 0$,
\[ \left| \E[\bX]{h(\bY)} \right| \leq \eps. \]
\end{lemma}

\begin{proof}[Proof of Lemma~\ref{cl:h-small}]
We note that, for any $z \in \R_{> 0}$,
\begin{align}
|h(z)| &= \left| (z-1)^{r} \sum_{\ell=1}^{\infty} (-1)^{\ell} \cdot \dfrac{(z-1)^{\ell}}{r + \ell} \right| \leq |z-1|^r \cdot \left| \sum_{\ell=1}^{\infty} (-1)^{\ell} \cdot \dfrac{(z-1)^{\ell}}{\ell} \right| \nonumber \\
	&= |z-1|^r \cdot |\log z| \lsim |z-1|^{r+1} + \left(\frac{9}{10}\right)^r + \ind\left\{ z \leq 1/10 \right\} \cdot \log(1/z). \label{eq:ub}
\end{align}
Hence, we have 
\[ \left| \E[\bX]{h(\bY)} \right| \leq \E[\bX]{\left| h(\bY)\right|} \leq \E[\bX]{ \left|\bY - 1\right|^{r+1}} + \eps / 2 + \E[\bX]{ \ind\left\{\bY \leq 1/10\right\} \log(1/\bY) },\]
where we used the fact that $\bY > 0$ and $r = \Theta(\log(1/\eps))$ to say $(9/10)^r < \eps/2$. In order to bound the above two quantities, we use the fact that the random variable $\bY$ is a subgamma random variable and thus has good concentration around its mean (which is $1$ for the case of $\bY$), giving the desired inequality. 

\begin{definition}[Subgamma Random Variable]
For $\sigma, B \in \R$, a random variable $\bZ$ with expectation $\mu$ is $(\sigma, B)$-subgamma if for all $\lambda \in \R$ with $|\lambda| < 1/|B|$, 
\[ \psi_{\bZ}(\lambda) \eqdef \ln \left( \E{e^{\lambda (\bZ - \mu)} } \right) \leq \dfrac{\lambda^2 \sigma^2}{2(1 - \lambda |B|)}. \]
\end{definition}

It is not hard to verify (see Section~\ref{sec:verifyYSubgamma}) that the random variable $\bY$ is centered at $1$, and that there are constants $\alpha, \beta \in \R_{\geq 0}$ so $\bY$ is $(\alpha/\sqrt{t}, \beta / t)$-subgamma. Then, by taking the Taylor expansion of $\E{e^{\lambda(Y - 1)}}$, we have that for any $|\lambda| < t / \beta$, and any $j \in \N$,
\begin{align}
\E[\bX]{ \left|\bY - 1\right|^{j}} &\leq \dfrac{j!}{\lambda^{j}} \cdot \exp\left( \dfrac{\alpha^2 \lambda^2}{2t (1 - \lambda \beta / t)}\right) \leq \dfrac{\alpha^{j} j!}{t^{j/2}} \cdot e^{3}, \label{eq:moment-bound}
\end{align}
by picking $\lambda = \sqrt{t} / \alpha$, which is less than $t / \beta$ for large enough $t$. Letting $j = r+1$ and setting $t = O(r^2)$, we get the desired bound of $o(\eps)$. In order to bound $\E[\bX]{\ind\left\{\bY \leq 1/10\right\} \log(1/\bY) } $, we compute it explicitly, and recall that $\bX \geq t$, so that the above event is satisfied only if $p_i \leq 1/10$.
\begin{align}
& \E[\bX]{ \ind\left\{\bY \leq 1/10\right\} \log(1/\bY) } \nonumber \\
&\leq \E[\bX]{ \dfrac{\ind\left\{ \bY \leq 1/10 \right\}}{\bY}} \label{eq:bound-1}\\
&= \sum_{\ell=t}^{t/(10p_i)} \binom{\ell-1}{t-1} p_i^t (1-p_i)^{\ell - t} \cdot \frac{t}{\ell p_i} \leq \frac{t}{10p_i} \max_{\ell \in [t, t/(10p_i)]} \left(\frac{e(\ell-1)}{t-1}\right)^{t-1} p_i^{t-1} \cdot \frac{t}{\ell} \nonumber \\
	&\leq \frac{t}{10} \max_{\ell \in [t, t/(10p_i)]} \left(\dfrac{e^2(\ell-1)}{t-1} \right)^{t-2} p_i^{t-2} = \exp(-\Omega(t)). \nonumber
\end{align}
\end{proof}

\subsubsection{Verifying $\bY$ is subgamma} \label{sec:verifyYSubgamma}

Recall that $\bX$ is the number of independent draws from a $\Ber(p)$ distribution until we see $t$ successes. In other words, we may express $\bX = \bX_1 + \dots + \bX_t$, where $\bX_i$ is the number of draws of $\Ber(p)$ before we get a single success. Then, we always satisfy
\[ \E{\bX_i} = \frac{1}{p} \qquad \Pr{\bX_i > \ell} = (1-p)^{\lceil \ell \rceil} < e^{-p\ell}\,.  \]
This, in turn, implies that for any $r \geq 1$
\[ \left( \E{ | \bX_i - 1/p|^r} \right)^{1/r} \leq \left( \E[\bX_i, \bX_i']{|\bX_i - \bX_{i}'|^r } \right)^{1/r} \leq 2 \left( \E{ |\bX_i|^r} \right)^{1/r} = O(r/p)\,,\]
where the first line is by Jensen's inequality, and the second is by the triangle inequality and H\"{o}lder inequality. Finally, we use the tail bound on $\bX_i$ to upper bound the expectation of $|\bX_i|^r$. Then, we have
\begin{align*}
\E{ e^{\lambda(\bX_i - 1/p)}} &= 1 + \lambda \E{ \bX_i - 1/p} + \sum_{k=2}^{\infty} \frac{\lambda^k}{k!} \cdot \E{|\bX_i - 1/p|} \\
	&= 1 + \sum_{k=2}^{\infty} \frac{\lambda^k}{k!} \left( O(k/p)\right)^k \leq 1 + O(\lambda^2 / p^2), \qquad \text{when $|\lambda|$ sufficiently smaller than $p$} \\
	&\leq \exp\left(O(\lambda^2/p^2)\right)
\end{align*}
Then, since $\bX_1,\dots, \bX_t$ are all independent, we have
\[ \E{ e^{\lambda(\bX - t/p)}} \leq \exp\left( O(\lambda^2 t / p^2)\right) \Longrightarrow \E{ e^{\lambda(\bY - 1)}} \leq \exp\left( O(\lambda^2/ t)\right),\]
and this bound is valid whenever $|\lambda|$ is sufficiently smaller than $t$.

\newcommand{\ignore}[1]{}

\section{Improving Sample Complexity via Bucketing}\label{sec:bucketing}
In this section, we focus on estimating the expected value of $\log(\bX/t)$ with error at most $\epsilon$. 
Our goal here is to remove the $\poly(\log k)$ dependencies in the sample complexity of estimation. In particular, we prove the following theorem, which improves on the dependence on $k$ in Theorem~\ref{thm:main-thm-simple}.
\begin{theorem}
There exists a single-pass data stream algorithm using $O(1)$ words of working memory that processes a stream of $O(k \log^4(1/\eps) / \eps^2)$ i.i.d. samples from an unknown distribution $\mathcal D$ on $[k]$ and returns an additive $\eps$ approximation of $H(\mathcal{D})$ with probability at least $2/3$. 
\end{theorem}

Given the work done in Section~\ref{sec:simple-alg}, it will suffice to estimate the quantity $H$ (we give the explicit reduction in Lemma~\ref{lem:reduce-to-H} shortly):
\begin{align}
    H &\coloneqq \E[\bi\sim\DD, \bX\sim \NB(t, p_{\bi})]{\log{(\bX/t)}}\,, \label{eq:val-H}
\end{align}
where $t$ is set to $\Theta(\log^2(1/\eps))$, such that we can then apply the correction term of Section~\ref{sec:simple-alg}. Recall that the randomness in the above expectation is taken over the random choice of $\bi \sim \DD$, and $\bX$ is a negative binomial random variable drawn from $\NB(t, p_{\bi})\,.$ First, we show that it suffices to estimate (\ref{eq:val-H}) in order to estimate the entropy, given our tools from Section~\ref{sec:simple-alg}.

\begin{lemma}\label{lem:reduce-to-H}
Consider a fixed distribution $\mathcal D$, and for $\eps > 0$ suppose $\hat{H} \in \R$ is such that $|H - \hat{H} | \leq \eps$. Then, there exists a $O(1)$ word streaming algorithm which given $\hat H$ and using an additional $O(\log(1/\eps) / \eps^2)$ independent samples from $\mathcal D$, outputs an estimate to the entropy of $\mathcal D$ up to error $\pm 2\eps$ with probability at least $0.9$.
\end{lemma}

\begin{proof}
The approach is to estimate 
\begin{align} 
\E[\bi \sim \mathcal D]{h_t(p_{\bi})} = \E[\bi \sim \mathcal{D}]{g(p_{\bi}, p_{\bi}^2,\dots, p_{\bi}^r)}. \label{eq:average-g}
\end{align}
There exists an algorithm using $O(\log(1/\eps) / \eps^2)$ samples to estimate the above quantity: for $j \in \{0, \dots, O(1/\eps^2)\}$, one takes a sample $\bi_j \sim \mathcal{D}$ and uses $r = O(\log(1/\eps))$ additional samples $\bs_1,\dots, \bs_r \sim \mathcal{D}$ to define
\[ \bB^{(j)}_m \eqdef \ind\{ \bs_1 = \dots = \bs_m = \bi_j \} \in \{0, 1\}, \]
and lets
\[ \bZ_j = g(\bB^{(j)}_1,\dots, \bB^{(j)}_r).\]
Then, let $\bZ$ be the average of all $\bZ_j$'s, which is an unbiased estimate to $\E[\bi \sim \mathcal{D}]{g(p_{\bi}, p_{\bi}^2,\dots, p_{\bi}^r)}$. Since $g$ is bounded (from Lemma~\ref{lem:bound-on-g}), the variance of $O(1/\eps^2)$ such values is a large constant factor smaller than $\eps^2$. By Chebyshev's inequality, we estimate (\ref{eq:average-g}) to error $\pm \eps$ with probability at least $0.9$. With that estimate, we will now use Lemma~\ref{lem:main-lemma}. Specifically, the entropy of $\mathcal{D}$ is exactly $\E[\bi \sim \mathcal{D}]{\log(1/p_{\bi})}$, and we have
\begin{align*}
\left| \E[\bi \sim \mathcal{D}]{\log(1/p_{\bi})} - \left(\hat{H} - \bZ \right)\right| &\leq \eps + \left| \E[\bi\sim\mathcal{D}]{\log(1/p_{\bi})} - \left(\hat{H} - \bZ\right) \right| \\
    &\leq \eps + \E[\bi \sim \mathcal{D}]{\left| \log\left(\frac{1}{p_{\bi}}\right) - \E{\boldeta_{\bi}}\right|} \leq 2\eps,
\end{align*}
where $\boldeta_{\bi}$ is the result of running $\SimpleEstimator(\mathcal{D}, \bi)$.
\end{proof}

 It thus suffices to design an algorithm to estimate (\ref{eq:val-H}). Our approach is to use a bucketing scheme. At a high level, we partition the range of $\bX$ into $L$ intervals: $I_1, I_2, \ldots, I_L$. We compute the conditional expectation of $\log(\bX/t)$ in each interval separately. Then, we take the weighted average of these conditional expectations, where the weights are determined by the probability of the intervals.

\paragraph{Unbounded $\pmb X$:}
As specified above, the random variable $\bX$ is a mixture of negative binomial random variables, so $\bX$ may be unbounded. In addition, if we had sampled $\bi \sim \DD$ where $p_{\bi}$ was very small, $\bX$'s value will tend to be very large.  
It will be convenient to introduce a parameter $\Xmax \in \N$ and consider the random variable $\bX' \coloneqq \min(\bX, \Xmax)$. 
Let $\tilH$ denotes the expected value of $\bX'$:
$$\tilH \coloneqq \E[\bi,\bX]{\log(\bX'/t)}\,.$$
For the rest of the section, we will seek to approximate $\tilH$, and the fact that this is a good estimate for $H$ follows from the following lemma. 
\begin{restatable}{lemma}{lemCutoffError}	
\label{lem:cutoffError}
Let $\bi \sim \DD$, and let $\bX$ and be a negative binomial random variable from $\NB(t, p_{\bi})$. Let $\bX'$ be the bounded version of $\bX$: $\bX' \coloneqq \min(\bX,\Xmax)$.  
Let $t \in \N$ and $\eps \in (0, 1)$. If we set $\Xmax = tk \ln(2)/\epsilon$, then
	$$ \left| H - \tilH \right| = \E[\bi, \bX]{\log(\bX/t) - \log(\bX'/t)}  \leq \epsilon\,.$$\end{restatable}

\begin{proof}
We note that since $\log(\cdot)$ is monotone increasing, we must have $H \geq \tilH$. To see that it is not much larger, note that we always have $\log z = \ln(z) / \ln(2)  \leq (z - 1) / \ln(2)$, which means
\begin{align*}
H - \tilH &= \E[\bi,\bX]{\log(\bX / \bX')} \leq \frac{1}{\ln(2)}\E[\bi, \bX]{\dfrac{\bX}{\min\{\bX, \Xmax\}} - 1} \leq \frac{1}{\ln(2)}\E[\bi, \bX]{\dfrac{\bX}{\Xmax}} \\
&= \frac{1}{\Xmax \cdot \ln(2)}\sum_{i=1}^k p_i \cdot \frac{t}{p_i} = \frac{tk}{\Xmax \cdot \ln(2)} \leq \eps.
\end{align*}
\end{proof}

\paragraph{Comparison to related work:} It is worth noting that the proofs in this section are inspired by the work of~\cite{AcharyaBIS19}. The authors used a similar bucketing technique to estimate entropy. While the structure of our proof is similar, there are subtle differences between our work and what they did. First, we are focusing on estimating different quantities. In particular, we work with an unbounded random variable while their estimator is bounded. Moreover,  they have a {two-step} bucketing system where they draw a sample $\bi$ and two estimates for $p_{\bi}$; they use one estimate for detecting which bucket falls into and the second one to estimate entropy in that bucket. One of the complications of this approach is that the second estimator may fall into a different bucket; Thus, they have to ``clip" the second estimator to make sure it is close to the bucket of the first estimator. We have circumvented these hurdles by using the same estimate for detecting which bucket we are in and estimating $\log(\bX/t)$ in that bucket. 

\paragraph{The algorithm:}
We write $\tilH$ in terms of conditional expectation in the intervals.
\begin{align*}
	\tilH = \sum_{\ell=1}^{L} \underbrace{\Pr[\bi\sim\DD, \bX\sim \NB(t, p_{\bi})]{\bX' \in I_\ell}}_{q_\ell \coloneqq} \cdot \underbrace{\E[\bi\sim\DD, \bX\sim \NB(t, p_{\bi})]{\log(\bX'/t) \mid \bX' \in I_\ell}}_{H_\ell \coloneqq }\,.
\end{align*}
Let $q_\ell$ denote the probability of $\bX'$ being in $I_\ell$, and $H_\ell$ denote the conditional expectation in $I_\ell$. Our algorithm estimate $q_\ell$ and $H_\ell$ for each interval to find an estimate for $\tilH$. Below we give a brief description of our algorithm, and the pseudocode can be found in Algorithm~\ref{alg:logEstimator}.

Below, we define $b_0 = t < b_1< \cdots < b_L=\Xmax$
to be $L+1$ parameters (which we will set shortly) that denote the boundary points of the intervals: 
$$I_{\ell} = [b_{\ell -1}, b_\ell) \quad\quad \forall i \in [L-1]\,,\quad \quad \quad I_L = [b_{L-1}, b_L]\,.$$

For each interval $I_\ell$, we draw $r_\ell$ samples from $\DD$, namely $\bi_1,\dots, \bi_{r_{\ell}} \sim \DD$. 
For each $\bi_{j}$, we start drawing samples from $\DD$ in the process of drawing a negative binomial random variable $\bX_j \sim \NB(t, p_{\bi_j})$; then, we will set $\bX_j' = \min(\bX_j, \Xmax)$. Furthermore, we will only consider $\bX_j'$'s that fall in $I_\ell$, which means that we can stop early if we already know $\bX_j'$ will be too large. In particular, if we draw $b_\ell$ samples and have not observed $t$ instances of $\bi_{j}$, we can already conclude $\bX_j'$ is not in $I_\ell$ and stop sampling.  Among these $r_\ell$ samples $\{ \bi_{1},\dots, \bi_{r_{\ell}} \}$, let $\bc_\ell$ denote the number of $\bX_j'$'s that fall into $I_\ell$. We estimate the weight of each bucket by $\hq \coloneqq \bc_\ell/r_\ell$. For the last bucket, we set $\hq_\ell$ in a way that the sum of the weight is one:
$$\hq_\ell = \frac{\bc_\ell}{r_\ell}\,, \,\quad \quad \forall j = 1, \ldots, L-1\,,\quad\quad  \hq_L \coloneqq 1 - \sum_{j=1}^{L-1} \hq_L\,.$$
Also, we compute an average of $\log (\bX_j'/t)$ of such $\bX_j'$'s and denote it by $\hH_\ell$:
$$\hH_\ell = \frac{\sum_{j=1}^{r_\ell} \ind\{\bX'_j \in I_\ell\} \cdot \log(\bX'_j/t)}{\bc_\ell}\quad\quad\quad \forall \ell = 1, \ldots, L\,.$$
In these definitions, we take $\hH_{\ell} = \log(b_{\ell}/t)$ if $\bc_{\ell} = 0$. Our estimate for $\tilH$ is the weighted sum of $\hH_\ell$:
$$\hH = \sum_{\ell=1}^L \hq_\ell \cdot \hH_\ell\,.$$

\begin{algorithm}[ht]
\caption{Estimating $\E{\log{\bX/t}}$ via Bucketing}
\label{alg:logEstimator}
\begin{algorithmic}[1]
\Procedure{LogEstimator}{$k$, $\epsilon$, sample access to $\DD$}
	\State{$\hH \gets 0$}
	\For{$\ell = 1, 2, \ldots, L$}
		\State{$\bc_\ell \gets 0, \ \hH_\ell \gets 0$}
		\For{$r_\ell$ times}
			\State{Draw $\bi \sim \DD$}
			\State{Draw $b_\ell$ samples from $\DD$ but terminate early if $t$ occurrences of $\bi$ are observed.}
			\State{$\bX \gets $ number of samples drawn}
            \State{$\hH_\ell  \gets \hH_\ell + \log(\bX/t) \cdot \ind\{\bX\in I_\ell\}$ and $\bc_\ell \gets \bc_\ell + \ind\{\bX\in I_\ell\}$}
		\EndFor		
    	\State{$\hH_\ell \gets \hH_\ell/c_\ell,\ \hq_\ell \gets c_\ell/r_\ell$}
    	\If {$\ell = L$}
    		\State{$\hq_L \gets 1 - \sum_{\ell=1}^{L-1} \hq_\ell$}
    	\EndIf
    	\State{$\hH \gets \hH + \hat q_\ell \hH_\ell $}
    \EndFor
\EndProcedure
\end{algorithmic}
\end{algorithm}

It is fairly straightforward to show that Algorithm~\ref{alg:logEstimator} uses a constant number of words. Note that $r_\ell$ (similarly $b_\ell$'s) can be computed from $r_{\ell-1}$, so we do not need to calculate and store all the $r_\ell$'s beforehand. Also, to compute $\hq_L$, we do not need all the $\hq_\ell$'s. We only need to keep a running sum of $\hq_\ell$\,. Thus, we only need a constant number of words of memory. We analyze correctness simply by bounding the variance. Namely, the remainder of the section will be devoted to proving the following lemma, which will imply that our estimator will be within $\pm \eps$ of $\tilH$ with constant probability.

\begin{lemma}\label{lem:square-deviation}
For any $k \in \N$ and $\eps > 0$, there exists a setting of $L > 0$, parameters $b_0 = t < b_1 < \dots < b_L = \Xmax$ and $r_{1},\dots, r_{\ell}$ for $\ell \in [L]$ such that 
\begin{align} 
\E{ \left( \sum_{\ell=1}^L \hq_{\ell} \cdot \hH_{\ell} - \sum_{\ell=1}^L q_{\ell} \cdot H_{\ell} \right)^2 } = O(\eps^2). \label{eq:square-deviation}
\end{align}
\end{lemma}

We note that, once we prove (\ref{lem:square-deviation}), we guarantee that our estimator is within $\pm O(\eps)$ of $\tilH$ with probability $0.9$ by Chebyshev's inequality. Before delving into the proof, we note that we may re-write the left-hand side of (\ref{eq:square-deviation}) as:
\begin{align}
\sum_{\ell=1}^L \hq_{\ell} \cdot \hH_{\ell} - \sum_{\ell=1}^L q_{\ell} \cdot H_{\ell} &= \sum_{\ell=1}^L \left(\hq_{\ell} - q_{\ell} \right) \hH_{\ell} + \sum_{\ell=1}^L q_{\ell} \left( \hH_{\ell} - H_{\ell}\right) \nonumber \\
	&= \sum_{\ell=1}^{L-1} \left( \hq_{\ell} - q_{\ell}\right) \hH_{\ell} + \left(1 - \sum_{\ell=1}^{L-1} \hq_{\ell} - 1 + \sum_{\ell=1}^{L-1} q_{\ell}\right) \hH_{L} + \sum_{\ell=1}^L q_{\ell} \left( \hH_{\ell} - H_{\ell}\right) \nonumber \\
	&= \sum_{\ell=1}^{L-1} \left(\hq_{\ell} - q_{\ell}\right) \left(\hH_{\ell} - \hH_{L}\right) + \sum_{\ell=1}^{L} q_{\ell} \left( \hH_{\ell} - H_{\ell}\right). \label{eq:divide-val}
\end{align}
Therefore, the upper bound on (\ref{eq:square-deviation}) follows from the following two lemmas.
\begin{lemma}\label{lem:first}
For any setting of $t = b_0 < \dots < b_{L} = X_{\max}$ and $\{ r_{\ell} \in \N : \ell \in [L]\}$, we have
\[ \E{ \left( \sum_{\ell=1}^L q_{\ell} \left(\hH_{\ell} - H_{\ell}  \right)\right)^2} \leq \sum_{\ell=1}^L q_{\ell} \left(1 - q_{\ell} \right)^{r_{\ell}} \cdot \log^2(b_{\ell} / b_{\ell-1}) + 2\sum_{\ell=1}^L \frac{\log^2(b_{\ell} / b_{\ell-1})}{r_{\ell} + 1}. \]
\end{lemma}

\begin{lemma}\label{lem:second}
For any setting of $t = b_0 < \dots < b_{L} = X_{\max}$ and $\{ r_{\ell} \in \N : \ell \in [L] \}$, we have
\[ \E{ \left( \sum_{\ell=1}^{L-1} (\hq_{\ell} - q_{\ell})(\hH_{\ell} - \hH_{L})\right)^2} \leq O(1) \cdot \sum_{\ell=1}^{L-1} \log^2(b_{L} / b_{\ell-1}) \cdot \log^{(\ell)} k\cdot \frac{q_{\ell} (1-q_{\ell})}{r_{\ell}}.\]
\end{lemma}

Given the above two lemmas, we can conclude the proof of Lemma~\ref{lem:square-deviation} assuming Lemma~\ref{lem:first} and Lemma~\ref{lem:second}. To upper bound (\ref{eq:square-deviation}), we first apply (\ref{eq:divide-val}) and then apply Lemma~\ref{lem:first} and Lemma~\ref{lem:second}.
\begin{align*}
&\E{ \left( \sum_{\ell=1}^L \hq_{\ell} \cdot \hH_{\ell} - \sum_{\ell=1}^L q_{\ell} \cdot H_{\ell} \right)^2 } = \E{ \left(\sum_{\ell=1}^{L-1} \left(\hq_{\ell} - q_{\ell}\right) \left(\hH_{\ell} - \hH_{L}\right) + \sum_{\ell=1}^{L} q_{\ell} \left( \hH_{\ell} - H_{\ell}\right)\right)^2} \\
&\qquad \qquad \leq 4 \cdot \E{  \left(\sum_{\ell=1}^{L-1} \left(\hq_{\ell} - q_{\ell}\right) \left(\hH_{\ell} - \hH_{L}\right) \right)^2 } + 4 \cdot \E{ \left(\sum_{\ell=1}^{L} q_{\ell} \left( \hH_{\ell} - H_{\ell}\right)\right)^2} \\
&\qquad\qquad \leq O(A + B + C),
\end{align*}
where we have
\begin{align*}
A &= \sum_{\ell=1}^L q_{\ell} (1 - q_{\ell})^{r_{\ell}} \cdot \log^2(b_{\ell} / b_{\ell-1}) \qquad\qquad B = \sum_{\ell=1}^L \dfrac{\log^2(b_{\ell} / b_{\ell-1})}{r_{\ell}+1} \\
C &= \sum_{\ell=1}^{L-1} \log^2(b_L / b_{\ell-1}) \cdot\log^{(\ell)} k \cdot \dfrac{q_{\ell} (1 - q_{\ell})}{r_{\ell}}
\end{align*}
Hence, we consider the following setting of parameters, where we let $L = \log^* k$,\footnote{Recall that $\log^{*} z$ is the number of iterated logarithms (base $2$) before the result is less than or equal to $1$.} such that we have $b_0 = t$, and
\begin{align}
\forall \ell \in \{1, \dots, L-1 \}, \quad b_{\ell} \eqdef \dfrac{tk}{(\log^{(\ell)} k)^4} \quad,\quad b_L \eqdef \dfrac{tk}{\eps} \quad\text{and let}\quad r_{\ell} \eqdef \dfrac{\log^2(b_{L} / b_{\ell-1}) \log^{(\ell)} k}{\eps^2}.  \label{eq:param-settings}
\end{align}

\begin{lemma}
For the above setting of parameters, $A$, $B$, and $C$ are at most $O(\eps^2)$.
\end{lemma}

\begin{proof}
We go through each of the above terms. For $C$, we have that $\sum_{\ell=1}^{L} q_{\ell} = 1$, so $C \leq \eps^2 \sum_{\ell=1}^{L-1} q_{\ell}(1 - q_{\ell}) \leq \eps^2$.
For $B$, we have
\begin{align}
B &\leq \sum_{\ell=1}^{L} \log(b_{\ell} / b_{\ell-1}) / r_{\ell} \leq \eps^2 \sum_{\ell=1}^L \dfrac{1}{\log(b_L / b_{\ell-1})}  \leq \eps^2 \sum_{h=1}^L \dfrac{1}{\log^{(L-h)} k}= O(\eps^2).
\end{align}
The last remaining thing to bound is $A$. Here, we consider a large constant $c_0$ and divide the set $[L]$ into $G = \{ \ell \in [L] : q_{\ell} \leq c_0\eps^2 / (\log^{2+0.01}(b_{\ell} / b_{\ell-1})) \}$ and $\overline{G} = [L] \setminus G$. Then, we have
\begin{align*}
A &\leq \sum_{\ell \in G} q_{\ell} \cdot \log^2(b_{\ell} / b_{\ell-1}) + \sum_{\ell \in \overline{G}} \left( 1 - q_{\ell}\right)^{r_{\ell}} \log^2(b_{\ell} / b_{\ell-1}) \\
	&\leq 2c_0\eps^2\sum_{\ell=1}^L\frac{1}{(\log^{(\ell)}k)^{0.01}} + \sum_{\ell \in \overline{G}} \exp\left(-q_{\ell} \cdot r_{\ell} \right) \cdot \log^2(b_{\ell} / b_{\ell-1})
\end{align*}
We note that the first summand on the right-hand side is also at most $O(\eps^2)$, by a similar argument to that of (\ref{eq:summing}). In particular, the sequence of iterated exponentiation grows quickly (looking at the summands in reverse order), so the summation of $1 / (\log^{(\ell)} k)^{0.01}$ is at most a constant. Finally, it remains to show the last part, which will also follow from the fact that the sum of iterated exponentiation converges, and the fact that $r_{\ell}$ has an additional dependence on $\log(1/\eps)$:
\begin{align*}
\sum_{\ell \in \overline{G}} \exp\left(- q_{\ell} \cdot r_{\ell} \right) \cdot \log^2(b_{\ell} / b_{\ell-1}) &\leq \sum_{\ell=1}^L \exp\left( - c_0(\log^{(\ell)} k)^{0.99} \log^2(1/\eps)\right) (\log^{(\ell)} k)^2 \leq \eps^{\Omega(c_0)}.
\end{align*}
\end{proof}

Before going on to prove Lemma~\ref{lem:first} and Lemma~\ref{lem:second}, we give a bound on the expected sample complexity.

\begin{lemma}
The expected sample complexity of the algorithm is  $O(k\log^4(1/\eps) / \eps^2)$.
\end{lemma}

\begin{proof}
For the intervals $\ell = \{ 1, \dots, L-1\}$, we always spend $r_{\ell}$ tries to determine whether a sample falls within a particular interval. Note that we take one sample to determine $\bi \sim \mathcal{D}$, and then we take at most $b_{\ell}$ samples. Therefore, the sample complexity for these is
\[ \sum_{\ell=1}^{L-1} r_{\ell} \cdot b_{\ell} = \dfrac{tk\log^2(1/\eps)}{\eps^2} \sum_{\ell=1}^{L-1} \dfrac{1}{\log^{(\ell)} k} \leq kt \cdot O(\log^2(1/\eps)/\eps^2), \]
where we used (and will continue to use) the fact that for any positive $\delta > 0$, we have
\begin{align}
 \sum_{\ell=1}^{L} \dfrac{1}{(\log^{(\ell)} k)^{\delta}} = \sum_{h=0}^{L} \dfrac{1}{(\log^{(L-h)} k)^{\delta}} = O(1). \label{eq:summing}
 \end{align}
In particular, for any $z > 1$ where $L = \log^* z$, 
\[ \log^{(L-h)} z > \underbrace{\exp\left(\exp\left(\dots \exp(1)\dots \right)\right)}_{\text{$h$ times}}. \]
In particular, the series in (\ref{eq:summing}) has denominators growing faster than geometric growth, so the sum is dominated by the largest value which is a constant. Finally, it remains to bound the expected sample complexity of the bucket $L$. Here, we note 
\[ r_L = \dfrac{O(1)}{\eps^2} \cdot \log^2\left( \frac{\log^{(L-1)} k}{\eps} \right) \leq O\left(\frac{\log^2(1/\eps)}{\eps^2}\right). \]
Therefore, the expected sample complexity incurred in the interval $L$ is
\[ r_{L} \cdot \sum_{i=1}^k q_i \cdot \dfrac{t}{q_i} = O(k\log^4(1/\eps) / \eps^2).\]
\end{proof}

We now give the proofs of the two lemmas.

\begin{proof}[Proof of Lemma~\ref{lem:first}]
Since the values of $q_1,\dots, q_{\ell}$ sum to $1$, we apply convexity of squaring and linearity of expectation to say
\begin{align*}
\E{ \left( \sum_{\ell=1}^L q_{\ell} \left( \hH_{\ell} - H_{\ell}\right)^2 \right)} \leq \sum_{\ell=1}^L q_{\ell} \E{ \left( \hH_{\ell} - H_{\ell}\right)^2} ,
\end{align*}
so it remains to bound the expected square of $\hH_{\ell} - H_{\ell}$. Recall the definition of $\hH_{\ell}$: we take a sample $\bi \sim \mathcal{D}$, then $\bX \sim \NB(t, p_{\bi})$ to check whether $\bX' \in I_{\ell}$; if so, we increment $\bc_{\ell}$ and use $\log(\bX' / t)$ as an estimate for $H_{\ell}$. Crucially, if we condition on any non-zero value of $\bc_{\ell}$, $\hH_{\ell}$ is an unbiased estimator given by averaging $\bc_{\ell}$ values, all of which are bounded in the interval $[\log(b_{\ell-1}/t), \log(b_{\ell} / t)]$. If $\bc_{\ell} = 0$, the estimator sets $\hH_{\ell}$ to $\log(b_{\ell}/t)$. In particular, since for all non-zero positive integers, $1/\bc \leq 2 / (\bc + 1)$, we may write
\begin{align*}
\E{ \left( \hH_{\ell} - H_{\ell} \right)^2 } &\leq \Pr{ \bc_{\ell} = 0 }\cdot \log^2(b_{\ell} / b_{\ell-1}) + 2 \E[\bc_{\ell}]{ \dfrac{(\log(b_{\ell}/t) - \log(b_{\ell-1} / t)^2}{\bc_{\ell} + 1} }.
\end{align*}
The first summand on the left-hand side is exactly $(1-q_{\ell})^{r_{\ell}} \log^2(b_{\ell} / b_{\ell-1})$. The second term is exactly
\[ 2 \log^2(b_{\ell}/ b_{\ell-1}) \E[\bc_{\ell}]{ \dfrac{1}{1+\bc_{\ell}}} \leq \dfrac{2\log^2(b_{\ell} / b_{\ell-1})}{q_{\ell} (r_{\ell}+1)}, \]
since $\bc_{\ell}$ is distributed as a binomial $\mathsf{Bin}(q_{\ell}, r_{\ell})$. The proof of this last inequality appears as Lemma~6 in \cite{AcharyaBIS19}, which we reproduce below for convenience:
\begin{align*}
\E[\bc \sim \mathsf{Bin}(q, r)]{ \dfrac{1}{1 + \bc}}&= \sum_{l=0}^{r} \binom{r}{l} q^{l}(1-q)^{r-l} \cdot \dfrac{1}{l+1} =  \sum_{l=0}^{r} \binom{r}{l} q^{l}(1-q)^{r-l} \cdot \dfrac{1}{r+1} \cdot \dfrac{r+1}{l+1} \\
	&= \dfrac{1 - (1 - q)^{r}}{q(r+1)} \leq \frac{1}{q(r+1)}.
\end{align*}
\end{proof}

\begin{proof}[Proof of Lemma~\ref{lem:second}]
Here, we write
\begin{align*}
\E{ \left( \sum_{\ell=1}^{L-1} (\hq_{\ell} - q_{\ell})(\hH_{\ell} - \hH_{L}) \right)^2} &\leq \E{ \left( \sum_{\ell=1}^{L-1} \dfrac{1}{\log^{(\ell)} k} \right)\left( \sum_{\ell=1}^{L-1} (\hq_{\ell} - q_{\ell})^2 (\hH_{\ell} - \hH_{L})^2 \cdot \log^{(\ell)} k \right)} \\
&\leq \left( \sum_{\ell=1}^{L-1} \dfrac{1}{\log^{(\ell)} k} \right) \sum_{\ell=1}^{L-1} \E{ (\hq_{\ell} - q_{\ell})^2 } \cdot \log^2(b_L / b_{\ell-1}) \log^{(\ell)} k
\end{align*}
by first multiplying and dividing by $\log^{(\ell)} k$, and then applying the Cauchy-Schwarz inequality. By the arguments we've already provided (see, for instance (\ref{eq:summing})), we have $\sum_{\ell=1}^{L-1} 1 /\log^{(\ell)} k = O(1)$, so it remains to bound the expectation square of $\hq_{\ell} - q_{\ell}$. However, we know that $\hq_{\ell}$ is the average of $r_{\ell}$ Bernoulli random variables, each set to $1$ with probability $q_{\ell}$. Hence, we have
\[ \E{ (\hq_{\ell} - q_{\ell})^2 } = \dfrac{q_{\ell}(1 - q_{\ell})}{r_{\ell}}. \]
\end{proof}

\ignore{\paragraph{Setting the parameters:} We set the number of intervals $L$ to be $\log^* (k)$.
That is, the number of times one needs to take $log$ in base two until it reaches a number that is at most one.
We set the values of $b_\ell$'s as follows:  
\begin{align*}
	& b_0 = t
	\\ & b_\ell = \frac{t\,k}{(\log^{(\ell)}(k))^3} \quad\quad\quad \ell \in [L-1]
	\\ & b_L = \Xmax = \frac{6\,t\,k}{\epsilon}\,.
\end{align*}
Here, $\log^{(\ell)}(.)$ indicates taking logarithm in base two $\ell$ times.
Moreover, we define $r_\ell$'s as follows:
\begin{align*}
	& r_\ell = \frac{117735\cdot\left(\log b_L - \log b_{\ell-1}\right)^{21/8}}{\epsilon^2}  \quad\quad\quad \ell \in [L-1]
	\\ & 
	r_L = r_{L-1}\,.
\end{align*}
Note that we have not spend any effort in optimizing a constant here

\begin{theorem}
	Given the parameters above, Algorithm~\ref{alg:logEstimator} uses $\Theta((k/\epsilon^2)\polylog(1/\epsilon))$ samples, and with  probability at least $0.95$, it output $\hH$ such that
	$$\left|\log(X/t) - \hH \right| \leq \epsilon\,.$$
	Moreover, one can implement it with constant many words of memory.  \amnote{I think the displayed equation above should have an expectation around the $\log$ or just use $H$}
\end{theorem}

\begin{proof}
	Our proof has three main parts: first, we analyze the error of our statistic. Second, we analyze the sample complexity of the algorithm. And at the end, we consider the memory usage of our algorithm. 

\paragraph{Error:} In this part, we analyze the error of our estimate. Our goal is to show that $|H - \hH|$ is at most $\epsilon$ with high probability. As the first step, we can write this error in terms of two terms: first, the error caused by replacing $X$ with $X'$; second, the error caused by using samples to estimate $\hq_\ell$'s and $\hH_\ell$'s.
\begin{align*}
	\left|H - \hH\right| & \leq |H - \tilH| + |\tilH - \hH| 
\end{align*}
In the following lemma, we show that by setting $\Xmax =6tk/\epsilon$, one can make sure the first term is at most $\epsilon/3$.

We focus on the second term: $\tilH - \hat H$. It is not hard to see that:
\begin{align*}
	\left|\tilH - \hH\right| & = \left|\E[i,X]{\log(X'/t)} - \hH \right|
	=\left|
	\sum_{\ell=1}^L q_\ell \cdot H_\ell - \hq_\ell \cdot \hH_\ell
	\right|
	 = \left| \sum_{\ell=1}^{L}(q_\ell \cdot H_\ell - q_\ell \cdot \hH_\ell) + (q_\ell \cdot \hH_\ell - \hq_\ell \cdot \hH_\ell)\right|
	\\ & \leq  \sum_{\ell=1}^{L}q_\ell \cdot \left|H_\ell -  \hH_\ell\right| + 
		\left| \sum_{\ell=1}^L (q_\ell - \hq_\ell) \cdot \hH_\ell \right|
\end{align*}
Now, we replace $q_L$ and $\hq_L$ by $q_L = 1-\sum_{\ell=1}^{L-1} q_\ell$ and $\hq_L = 1-\sum_{\ell=1}^{L-1} \hq_\ell$ respectively. Note that using this estimator for $\hq_L$, instead of estimating $\hq_L$ by $c_L/r_L$ plays a crucial role in this proof (as well as the similar proof in~\cite{AcharyaBIS19}).  By this simple replacement, we get an error term of the form $\left|q_\ell - \hq_\ell \right|\cdot (\log b_L - \log b_{\ell-1})$ below. It implies that we can estimate $\hq_\ell$ with lower accuracy as $\ell$ gets larger. And, this is exactly where the {\em early stop} of sampling  comes into play. Estimating $\hq_\ell$ requires a lot less samples for small $\ell$, since we would only need to check if $X'$ is larger than $b_\ell$. Therefore, we can afford larger number of trials, $r_\ell$, for smaller $\ell$'s.
\begin{align*}
	\left|\tilH - \hH\right| 
	& = \sum_{\ell=1}^{L}q_\ell \cdot \left|H_\ell -  \hH_\ell\right| + 
	\left|
	\sum_{\ell=1}^{L-1} 
	(q_\ell - \hq_\ell) \cdot \hH_\ell + \left(1-\sum_{\ell=1}^{L-1} q_\ell\right) \cdot \hH_L - \left(1-\sum_{\ell=1}^{L-1} \hq_\ell\right) \cdot \hH_L
	\right|
	\\ 
	& \leq \sum_{\ell=1}^{L} q_\ell \cdot \left|H_\ell -  \hH_\ell\right| + \sum_{\ell=1}^{L-1} 
	\left|q_\ell - \hq_\ell \right|\cdot \left|\hH_L - \hH_\ell\right| 
	\\ 
	& \leq \sum_{\ell=1}^{L} q_\ell \cdot \left|H_\ell -  \hH_\ell\right| + \sum_{\ell=1}^{L-1} 
	\left|q_\ell - \hq_\ell \right|\cdot (\log b_L - \log b_{\ell-1}) 
\end{align*}

Next, we show that with a probability of at least 0.99, the two terms in the last line are at most $\epsilon/3$. Given these bounds, we can conclude that $\left|H - \hH\right|$ is at most $\epsilon$ with high probability as desired. 

For the first error term, we prove:
\begin{equation}\label{eq:termOne}
	\Pr{\sum_{\ell=1}^{L}q_\ell \cdot \left|H_\ell -  \hH_\ell\right|  \geq \epsilon/3} < 0.01\,.
\end{equation}
We start by proving a concentration bound for $\hH_\ell$. Acharya et al. have shown a randomized version of Hoeffding inequality which the number of samples to estimate a quantity comes from a binomial distribution. We stated this lemma below. See Lemma~3 in~\cite{AcharyaBIS19}.
\begin{lemma}[Random Hoeffding's Inequality~\cite{AcharyaBIS19}]
Let $C \sim\bin(r, q)$ be a binomial random variable with success probability $q$ and $r$ trials. Let $Y_1, Y_2, \ldots, Y_C$ be $C$ independent random variables in the range $[a, b]$. Let $\overline{Y} = (Y_1 + \cdots + Y_C)/C$ denote the average of these $C$ random variables.   Then, we have:
	 $$\Pr{q \cdot \left|\overline{Y} - \E{\overline{Y}}\right| \geq s} \leq 3 \exp\left(-\frac{r\,s^2}{8\,q\,(b-a)^2}\right)\,.$$
\end{lemma}
We use the setting of this lemma to show the concentration of $\hH_\ell$. Note that using our bucketing scheme our random variables are between $b_{\ell-1}$ and $b_\ell$. And, we get a sample in interval $\ell$ with probability $q_\ell$. Thus, we get:
$$\Pr{q_\ell \cdot \left|\hH_\ell - H_\ell\right| \geq s_\ell} \leq 3\exp\left(-\frac{r_\ell\,s_\ell^2 }{8 \,q_\ell (\log b_\ell - \log b_{\ell-1})^2}\right) \leq 3\exp\left(-\frac{r_\ell\,s_\ell^2 }{8 (\log b_\ell - \log b_{\ell-1})^2}\right)\,.$$

We have the above bounds for each $\ell$. However, our goal is to find a sequence of $s_\ell$'s to obtain a concentration bound for the sum as in Equation~\eqref{eq:termOne}. 
We define the $r_\ell$'s to be large enough, so we can use Lemma~\ref{lem:Svals} which states there exists a 
sequence $\{s_\ell \}_{\ell=1}^L$\,, with the following properties:

\begin{enumerate}
	\item[$i)$] $\sum_{\ell=1}^L s_\ell \leq \frac{\epsilon}{3}\,,$
	\item[$ii)$] $
	\sum_{\ell=1}^{L} 3\exp\left(-\frac{r_\ell\,s_\ell^2 }{8 \, (\log b_\ell - \log b_{\ell-1})^2}\right) \leq 0.01\,.$
\end{enumerate}
It is not hard to see that one can drive Equation~\eqref{eq:termOne} using the above properties and the union bound:
\begin{align*}\Pr{\sum_{\ell=1}^L q_\ell \cdot \left|\hH_\ell - H_\ell\right| \geq \epsilon/3}   & \leq \Pr{\sum_{\ell=1}^L q_\ell \cdot \left|\hH_\ell - H_\ell\right| \geq \sum_{\ell=1}^L s_\ell}  
\\ & \leq \Pr{\exists \ell \in[L]: q_\ell \cdot \left|\hH_\ell - H_\ell\right| \geq s_\ell} 
\\ & \leq \sum_{\ell=1}^L 3\exp\left(-\frac{r_\ell\,s_\ell^2 }{8 (\log b_\ell - \log b_{\ell-1})^2}\right) \leq 0.01
\,.
\end{align*}

Now, we focus on the second error term. We aim to prove:
\begin{equation}\label{eq:termTwo}
	\Pr{\sum_{\ell=1}^{L-1} 
	\left|q_\ell - \hq_\ell \right|\cdot (\log b_L - \log b_{\ell-1})    \geq \epsilon/3} \leq 0.01
\end{equation}

Recall that we use a $r_\ell$ samples per bucket to estimate the weight of each bucket. We use the Hoeffding's inequality to show $\hq_\ell = c_\ell/r_\ell$ is close to its expectation, $q_\ell$. For every $\ell \in [L-1]$ and $t_\ell \in [0,1]$, we have:
\begin{align*}
	\Pr{\left|\frac{c_\ell}{r_\ell} - q_\ell\right| \leq t_\ell} \leq 2\exp\left(-2\,r_\ell \,t_\ell^2\right)
	\,.
\end{align*}

Given the values of $r_\ell$'s, we can use Lemma~\ref{lem:Tvals}, and find a sequence of $\{t_\ell\}_{\ell=1}^{L-1}$ with the following properties:
\begin{enumerate}
	\item[$i)$] $\sum_{\ell=1}^{L-1} t_\ell \cdot (\log b_L - \log b_{\ell-1}) \leq \epsilon/3\,,$
	\item[$ii)$] $
	\sum_{\ell=1}^{L-1} 2\exp\left(-2\,r_\ell \, t_\ell^2\right) \leq 0.01\,.$
\end{enumerate}

Now, with all these ingredients and the union bound, we can prove Equation~\eqref{eq:termTwo}:
\begin{align*}
	\Pr{\sum_{\ell=1}^{L-1} 
	\left|q_\ell - \hq_\ell \right|\cdot (\log b_L - \log b_{\ell-1})    \geq \epsilon/3}  & \leq  \Pr{\sum_{\ell=1}^{L-1} 
	\left|q_\ell - \hq_\ell \right|\cdot (\log b_L - \log b_{\ell-1})    \geq \sum_{\ell = 1}^{L-1} t_\ell \cdot (\log b_L - \log b_{\ell-1})} 
	\\ & \leq \Pr{\exists \ell \in[L-1]: \left|q_\ell - \hq_\ell \right|  \geq t_\ell}
	\\ & \leq \Pr{\exists \ell \in[L-1]: \left|q_\ell - c_\ell/r_\ell \right|  \geq t_\ell} 
	\\ & \leq \sum_{\ell=1}^{L-1} 2\exp\left(-2\,r_\ell \, t_\ell^2\right) \leq 0.01\,.
\end{align*}
Thus, by the union bound with probability at least 0.98 $\left|H - \hH\right|$ is at most $\epsilon$.

\paragraph{Number of samples:}
Note that we have $L = \Theta(\log^*k)$ many buckets. For each bucket, we draw $r_\ell$ samples, and we draw $X'$ using at most $b_\ell$ samples. For the very last bucket, instead of using $b_L$ as an upper bound for the number of samples, we use the upper bound on the expected number of samples we need to draw a negative binomial random variable; For a random $i \sim \DD$, we need $tk$ samples in expectation to draw $X$ (as well as for $X'$). Putting it all together, we have:
\begin{align*}
	\E{\text{\# samples}} & \leq \sum_{\ell=1}^{L-1} b_\ell \cdot r_\ell + t\, k\, r_L
	\\& = \sum_{\ell=1}^{L-1} \frac{t\,k}{(\log^{(\ell)}(k))^3} \cdot \Theta \left(\frac{\left(\log b_L - \log b_{\ell-1}\right)^{21/8}}{\epsilon^2}\right)  + t\, k\, r_L
	\,.
\end{align*}
Note that, we have:
\begin{align*}
	\log b_L - \log b_{\ell-1} & = \log \left(b_L/b_{\ell-1}\right) = \log \frac{(\log^{(\ell-1)}(k))^3}{\epsilon}=3 \, \log^{(\ell)}(k) + \log (1/\epsilon)\,. 
\end{align*}
Therefore, we can bound the expected number of samples as follows via Proposition~\ref{prop:convergentSeries}:
\begin{align*}
	\E{\text{\# samples}} & \leq 
	\frac{t\,k}{\epsilon^2} \cdot \Theta \left(
	\sum_{\ell=1}^{L-1} \frac{\left(\log^{(\ell)}(k)\right)^{21/8}}{(\log^{(\ell)}(k))^3}  + (\log(1/\epsilon))^{21/8}
	\right) + t\,k\cdot \Theta\left((\log^{(L)}(k))^{21/8} + (\log(1/\epsilon))^{21/8}\right)
	\\ &= \Theta(\polylog(1/\epsilon) (k/\epsilon^2))
	\,.
\end{align*}
Note that using Markov's inequality, we can assume that with probability at most 0.01, we would use more than $100\E{\text{\# samples}}$ samples. Thus, the sample complexity remains $\Theta(\polylog(1/\epsilon) (k/\epsilon^2))$ even in the worse case.

\end{proof}
}

\section{Conclusions}
We presented an algorithm for returning an additive $\epsilon$ approximation of the Shannon entropy of a distribution over $[k]$. The algorithm required $O(k\,\epsilon^{-2} \log^4 (1/\epsilon) )$ i.i.d.~samples from the unknown distribution and a constant number of words of memory. In terms of the $\epsilon$ dependence, this improves over the state-of-the-art \citep{AcharyaBIS19} by a factor $1/\epsilon$ in the sample complexity. More generally, we expect that the technique used, that of correcting the bias via low-degree polynomials will be useful in the context of the other inference problems in the data stream setting.

The main open problem is determining whether the sample complexity of our result is optimal. We conjecture that this is the case, up to the $\poly(\log(1/\eps))$ factors. Recall that without a memory constraint the sample complexity is known to be  $n = \Theta(\max\{ \epsilon^{-1} \cdot  k/\log(k/\epsilon), \epsilon^{-2} \log^2 k\})$ \citep{ValiantV17,ValiantV11,JiaoVHW15,WuY16}. To prove  a $\Omega(k/\epsilon^2)$ lower bound for the memory constrained version, we conjecture the following randomized process can be used to generate distributions over $[2k]$ that look alike to any constant space algorithm that uses $o(k/\epsilon^2)$ samples but they have {\em different } entropies. 

Suppose we have $k$ Bernoulli random variables with parameter $\alpha$: $Y_1, \ldots, Y_k$. And, we have $k$ Rademacher random variables $Z_1, \ldots, Z_k$ (that are $+1$ or $-1$ with probability $1/2$). We construct distribution $p$ in such a way that it is uniform over $k$ pairs of elements $(1, 2), (3, 4), \ldots, (2k-1, 2k)$. However, conditioning on pair $(2i-1, 2i)$, we may have a constant bias based on the random variable $Y_i$. And, we decide about the direction of the bias based on $Z_i$. More precisely, we set the probabilities in $p$ as follows: 
$$p_{2i-1} = \frac{1 + Y_i\cdot Z_i/4}{2k}\,,\quad \quad p_{2i} = \frac{1 - Y_i\cdot Z_i/4}{2k}\quad\quad\quad \forall i \in [k]\,.$$
Now, it is not hard to show that if we generate two distributions as above with $\alpha = (1 + \epsilon)/2$ and $\alpha = (1-\epsilon)/2$, then their entropies are $\Theta(\epsilon)$ separated with a constant probability. Thus, any algorithm that can estimate the entropy has to {\em distinguish}  $\alpha = (1 + \epsilon)/2$ from $\alpha = (1-\epsilon)/2$.  Intuitively, 
to learn $\alpha$, we would require to {\em determine} $\Omega(1/\epsilon^2)$ many of $Y_i$'s. Since we have only a constant words of memory, we cannot perform the estimation of the $Y_i$'s in parallels. Thus, any natural algorithm would require to draw $\Omega(k/\epsilon^2)$ samples.

\appendix

\section{Variance of $\SimpleEstimator$}

We now bound the variance of our estimator by $O(\log^2 k)$. Recall that the output of $\SimpleEstimator$ is given by $\log(\bX / t) - g(\bB_1,\dots, \bB_r)$, where the function $g$ is bounded. Since the variance we seek is $O(\log^2 k)$, it suffices to show that the variance of $\log(\bX / t)$ is $O(\log^2 k)$ with $i\sim \mathcal{D}$, since subtracting $g$ changes the estimate by at most a constant (see Lemma~\ref{lem:bound-on-g}).

\begin{lemma}\label{lem:tail2}
Let $\bi \sim \mathcal{D}$ and $\bX$ denote the number of independent trials from $\Ber(p_{\bi})$ before we see $t$ successes. Then, $\Var{\log (\bX/t)}=O(\log^2 k)$.
\end{lemma} 

\begin{proof}
Let $\Xmax = 2 k t$, and consider the random variable $\bX' = \min \{ \bX, \Xmax \}$. Then 
\begin{align*}
\Var{\log(\bX / t)} &\leq \E{\left(\log(\bX/t) - \log(\bX' / t) + \log(\bX'/t) \right)^2}\\ 
    &\leq 2 \cdot \E{\left(\log(\bX/t) - \log(\bX' / t) \right)^2} + 2 \cdot \E{\log^2(\bX'/t)} \\
&\leq 2 \cdot \E{\log^2(\bX / \bX')} + 2\log^2(2k)\\
&\leq \frac{4}{\ln^2(2)} \cdot \E{\left(\sqrt{\frac{\bX}{\bX'} - 1}\right)^2} + 2 \log^2(2k),
\end{align*}
where we used that $\log(\bX' / t) \leq \log(2k)$ always, and that $\log(z) \leq \sqrt{z - 1} / \ln(2)$ for all $z \geq 1$. Then, 
\begin{align*}
\E{\dfrac{\bX}{\bX'} - 1} \leq \E{\dfrac{\bX}{\Xmax}} = \frac{1}{\Xmax} \sum_{i=1}^k p_i \cdot \frac{t}{p_i} = \frac{tk}{\Xmax} = 2.
\end{align*}
\end{proof}

\ignore{

\begin{proof}
For any $q \in [0, 1]$, let $A_q$ be the event the sample being tracked has probability $q$. Let $\tau_q=100t/q$. Then
\begin{align*}
\E{\log^2 (\bX/t)|A_q}&=  
\E{\ind\{\bX\leq \tau_q\} \log^2 (\bX/t)|A_q}+\E{\ind\{\bX\geq \tau_q\} \log^2 (\bX/t)|A_q} \\
&\leq  \log^2 (3/q) + \E{\ind\{\bX\geq \tau_q\} (\bX/ t)|A_q}
\end{align*}
We now analyze the second term:
\begin{eqnarray*}
\E{\ind\{\bX\geq \tau_q\} (\bX/t)|A_q}&= &\frac{1}{t}\sum_{i=\tau_q}^\infty \binom{i-1}{t-1} q^t (1-q)^{i-t} \cdot i 
\\ &=&   \frac{1}{t}
\binom{\tau_q-1}{t-1} q^t (1-q)^{\tau_q-t} \tau_q (1+\alpha_1+\alpha_2+\ldots )
\end{eqnarray*}
where 
\begin{eqnarray*}
\alpha_j &=&
\frac{\binom{\tau_q+j-1}{t-1}}{\binom{\tau-1}{t-1}} (1-q)^j \frac{\tau_q+j}{\tau_q}
\\ &= &
\frac{(\tau_q+j-1)\ldots (\tau_q)}{(\tau_q+j-t)\ldots (\tau_q-t+1)} (1-q)^j \frac{\tau_q+j}{\tau_q}
\\ &\leq & \left ( \frac{\tau_q}{\tau_q-t+1} \right )^j ((\tau_q+1)/\tau_q)^j (1-q)^j 
\\ &=&
\left ( \frac{\tau_q+1}{\tau_q-t+1} \right )^j (1-q)^j \\ 
&\leq &
\left ( 1+t/(\tau_q-t) \right )^j (1-q)^j \\
&\leq &
\left ( 1+q/2 \right )^j (1-q)^j \leq (1-q/2)^j \ .
\end{eqnarray*}
Therefore $1+\alpha_1+\alpha_2+\ldots \leq 2/q$ and so 
\begin{eqnarray*}
\E{\ind\{\bX\geq \tau_q\} \bX|A_q} \leq (2/q) \cdot \binom{\tau_q-1}{t-1} q^t (1-q)^{\tau_q-t} \tau_q
& \leq & 2 \cdot \left ( \frac{eq(\tau_q-1)}{t-1}\right )^{t-1} e^{-99t} \tau_q \\
& \leq & 2 \cdot (101 e)^{t-1} e^{-99t} \tau_q = O(1/q) \ . 
\end{eqnarray*}
Therefore, $\E{\log^2 (\bX/t)}\leq \sum_{i=1}^k p_i (\log^2 (3/p_i)+O(t/p_i))=O(\log^2 k)$.
\end{proof}
}

\section{Bit Complexity of Storing $g$ in $\SimpleEstimator$}\label{sec:bit-complexity} 

We verify the bit complexity of storing $g$. Then, notice that the algorithm needs to store the parameter $t$, an integer counter between $0$ and $t$, and a number $\bX$ which is at most the sample complexity (which is bounded in expectation by Fact~\ref{fact:sample-complexity}).

We turn to verifying that $g$ may be stored with bounded bit-complexity. Recall that, for $\eps \in (0,1)$, we set $t = \Theta(\log^2(1/\eps))$ and $r = \Theta(\log(1/\eps))$. Recall that $g$ is the linear function given by having $g(z, z^2, \dots, z^r)$ be the degree-$r$ Taylor expansion of $\log z$ centered at $1$. In order to show the coefficients of $g$ may be stored with bounded bit-complexity, we compute the degree-$r$ Taylor expansion of $\log z$ at $1$:
\begin{align*}
\sum_{i=1}^r \dfrac{(-1)^{i+1}}{i} \cdot (z - 1)^{i} &= \sum_{i=1}^r \sum_{j=0}^i \binom{i}{j} \cdot \frac{(-1)^{j-1}}{i} \cdot  z^j \\
& = \left(-\sum_{i=1}^r \binom{i}{j} \cdot \frac{1}{i} \right) + \sum_{j=1}^r z^j \left( \sum_{i = j}^{r} \binom{i}{j} \dfrac{(-1)^{j-1}}{i}\right).
\end{align*}
For any $j \in \{ 1, \dots, r\}$, the $j$-th coefficient of $g$ may be re-written as:
\begin{align*}
\sum_{i = j}^{r} \binom{i}{j} \dfrac{(-1)^{j-1}}{i} &= \dfrac{1}{r!} \sum_{i=j}^{r} \binom{i}{j} (-1)^{j-1} \cdot \frac{r!}{i},
\end{align*}
and $r! / i$ is always an integer at most $r^r$, and $\binom{i}{j}$ is an integer bounded by $r^r$; since we are summing at most $r$ such values, the coefficients have bit-complexity $O(r \log r)$. Similarly, a bound of $O(r \log r)$ on the bit-complexity of the $0$-th coefficient of $g$ follows from a similar argument.

\ignore{

\section{Technical Lemmas for Section \ref{sec:bucekting}}\label{app:bucketing-proofs}

\lemCutoffError*

\begin{proof}
By total law of probability, we have: 
\begin{equation}\label{eq:XXdiffmain}
	\begin{split}
		\E{\log(X/t) - \log(X'/t)} 
	& = \sum_{i\in[k]} p_i \cdot \left( \E{\log(X/t) - \log(X'/t) \mid i}\right)
	\end{split}
\end{equation}
For now, we fix $i$ and assume $X$ is a negative binomial random variable from $\NB(t, p_i)$.
We use the integral identity to compute the expectation:
\begin{equation}\label{eq:XXprimediffBound}
	\begin{split}
		\E[X \sim \NB(t, p_i)]{\log (X/t) - \log (X'/t)} & = \E[X \sim \NB(t, p_i)]{\log X - \log X'} 
		\\& = \int_{z\geq0} \Pr{\log X - \log X' > z}\, dz
		\\ & = \int_{z\geq0} \Pr{X > e^z X'}\, dz
		\\ & = \int_{z\geq0} \Pr{X > e^z \Xmax}\, dz
		\,.
	\end{split}
\end{equation}
Note that last equality is due to the following: The event of $X>e^z X'$ implies that $X > X'$; On the other hand, $X'$ is the minimum of $X$ and $\Xmax$. Thus, if $X' < X$, then it must be the case that $X'$ is equal to $\Xmax$. Next, we define a threshold $T$ as follows:
$$T \coloneqq \max\left(\ln \left(\frac{2\,\E{X}}{\Xmax}\right), 0\right) \,,$$
where $\E{X} = t/p_i$. We separate the above integral based on the threshold $T$.

\begin{equation}\label{eq:XXdiffcnt}
	\begin{split}
		& \E[X \sim \NB(t, p_i)]{\log (X/t) - \log (X'/t)}
		= 
		\int_{0}^{T} \Pr{X \geq e^z\Xmax}\, dz + \int_{z\geq T} \Pr{X \geq e^z\Xmax}\, dz
		\\ & \quad \quad 
		= T + \int_{z\geq T} \Pr{X \geq e^z\Xmax}\, dz
		\\ & \quad \quad 
		= T +  \int_{z\geq T} \Pr{X \geq e^{z - T}\max\left(e^{2\E{X}/\Xmax}, 1\right) \cdot \Xmax}\, dz
		\\ & \quad \quad 
		= T +  \int_{z\geq T} \Pr{X \geq e^{z - T}\max\left(2\E{X}, \Xmax \right)}\, dz
		\\ & \quad \quad 
		\leq T +  \int_{z\geq T} \Pr{X \geq e^{z - T}2\E{X}}\, dz
		\\ & \quad \quad 
		= 
		T +  \int_{u\geq 0} \Pr{X \geq 2\,e^{u}\,\E{X}}\, du
	\end{split}
\end{equation}
Observe that this separation let us 
decompose the expectation  in two terms: the first one only depends on how far $\E{X}$ is from $\Xmax$, and the second one depends on how $X$ may land far from its expectation. Now, to bound the probability in the second term, we use the tail bounds of negative binomials shown in Lemma~\ref{lem:tialBoundNB}. Since $2\,e^u$ is at least two, we have:
\begin{align*}
	\Pr{X \geq 2\,e^{u}\,\E{X}} \leq \exp\left(-t\,e^u /4\right)\,.
\end{align*}
Thus, we get:
\begin{equation*}
	\begin{split}
		\int_{u\geq0} \exp\left(-\left(t/4\right) \cdot e^u\right) \, du
		=\Gamma(0,t/4)
		\,,
	\end{split}
\end{equation*}
where $\Gamma(.,.)$ denotes the incomplete gamma function. The value of the gamma function when the first parameter is zero is bounded as follows
according to \cite{pinelis2020exact}:
\begin{align*}
	\Gamma(0, x) \leq e^{-x} \ln\left(\frac{x+1}{x}\right)\,.
\end{align*}
Since $ t/4 \geq 1/4$, $\ln((x+1)/x)$ is at most two. Therefore, we get:
\begin{align*}
	\int_{u\geq0} \exp\left(-\left(t/4\right) \cdot e^u\right) \, du
		=\Gamma(0,t/4) \leq 2 e^{-t/4} \leq \epsilon/2\,,
\end{align*}
where the last line holds for $t\geq 4\ln(4/\epsilon)$.
Continuing our calculation in Equation~\eqref{eq:XXdiffcnt}, we get:
\begin{align*}  
	\E[X \sim \NB(t, p_i)]{\log (X/t) - \log (X'/t)} 
	& \leq T + \epsilon/2 = \max\left(\ln \left(\frac{2\,t}{p_i\,\Xmax}\right), 0\right) + \epsilon/2\,.
\end{align*}
Now, going back to Equation~\eqref{eq:XXdiffmain}, we obtain:
\begin{align*}
	\E{\log (X/t) - \log (X'/t)} 
	& \leq \sum_{i=1}^k p_i \cdot \left(\max\left(\ln \left(\frac{2\,t}{p_i\,\Xmax}\right), 0\right) + \epsilon/2\right)
	\\ & \leq
	\epsilon/2 + \sum_{i: p_i \leq 2t/\Xmax} p_i \cdot \left(\ln(1/p_i) - \ln(\Xmax/2t)\right)
	\,.
\end{align*}
Let $w$ denotes the total probability of the elements with probability less than $2\,t/\Xmax$, and let $v$ denotes the number of such elements. 
$$w \coloneqq  \sum_{i: p_i \leq 2t/\Xmax} p_i\,, \quad\quad\quad v \coloneqq \left|\{i\mid p_i \leq 2t/\Xmax\}\right|\,.$$ 
By concavity of the natural log function and  Jensen's inequality, 
we have:
$$\sum_{i: p_i \leq 2t/\Xmax} p_i \ln (1/p_i) \leq w\ln(v/w) \leq w\ln(k/w)\,.$$

Now, we set $\Xmax$ to $2tk/\epsilon$. Thus, we get:

\begin{align*}  
	\E{\log (X/t) - \log (X'/t)}
	& \leq \epsilon/2 + w\ln(k/w) - w\ln(k/\epsilon)
	\\ & \leq \epsilon/2 + w\ln(\epsilon/w)
	\,.
\end{align*}
Note that by definition, $w$ is at most $k \cdot (2t/\Xmax)\leq \epsilon/2$. Using that $w\ln(\epsilon/(2w))$ is a concave function, and its maximum happens when $w = \epsilon/2 e$, one can conclude that
\begin{align*}  
	\E{\log (X/t) - \log (X'/t)}
	& \leq \epsilon\,.
\end{align*}
\end{proof}

\begin{lemma}\label{lem:tialBoundNB}
	Let $X$ be a negative binomial random variable drawn from $\NB(t, p)$. Then, we have the following tail bounds: 
\begin{align*}
	&
	\Pr{X \geq \tau t/p} \leq \exp\left(- \frac{t\tau}{8}\right)&\quad\quad\quad \forall \tau \geq 2\,,
	\\ &
	\Pr{X \leq \tau t/p} \leq \min\left(\exp\left(- \frac{t}{12}\right)\,,\tau\right)&\quad\quad\quad \forall \tau \leq 1/2\,.
\end{align*}
\end{lemma}
\begin{proof}
Note that one can view the event of $X$ being at least  $\tau t/p$ in a different way. This event implies out of $r \coloneqq \ceil{\tau t/p}$ Bernoulli trial with probability $p$ at there are at most $t$ successes. Let $Y$ denote the number of successes in $r$ Bernoulli trials. By the Chernoff bound, we obtain: 
\begin{align*}
	\Pr{X \geq \tau t/p} &\leq \Pr{Y \leq t} = \Pr{\frac{Y}{r} \leq \frac{t}{r}} \leq \Pr{\frac{Y}{r} \leq p/\tau =  (1-(1-1/\tau))\,p}
	\\& \leq \exp\left(- \frac{rp(1-1/\tau)^2}{2}\right) \leq \exp\left(- \frac{t \,\tau (1-1/\tau)^2}{2}\right)\,.
\end{align*}
To conclude the first tail bound from the above inequality, we use that $1-1/\tau$ is at least 1/2 for $\tau \geq 2$. 

Similar to the first part, we view the event of $X$ being smaller than $\tau t/p$ as an indication where among $r' = \floor{\tau t/p}$ trials of Bernoulli random variable, there are at least $t$ successes. Let $Y'$ indicate the number of successes in this case. Thus, using the Chernoff bound, we get: 
\begin{align*}
	\Pr{X \leq \tau t/p} &\leq \Pr{Y' \geq t} = \Pr{\frac{Y'}{r'} \geq \frac{t}{r'}} \leq \Pr{\frac{Y'}{r'} \geq \frac{p}{\tau} = ((1/\tau - 1) + 1)\,p}
	\\ & \leq \exp\left(-\frac{r'p (1/\tau - 1)^2}{1/\tau + 1}\right) \leq \exp \left(-\frac{t\,\tau(1/\tau - 1)^2}{2(1/\tau + 1)}\right) \leq \exp\left(-\frac{t}{12}\right)\,.
\end{align*}
Note that the last line is due to the fact that for $\tau \leq 0.5$, we have:
$$\frac{\tau(1/\tau - 1)^2}{2(1/\tau + 1)} \geq \frac{1}{12}\,.$$

In addition, if we use Markov's inequality:
\begin{align*}
	\Pr{X \leq \tau t/p} &\leq \Pr{Y' \geq t} = \Pr{\frac{Y'}{r'} \geq \frac{t}{r'}} \leq \frac{p\,r'}{t} \leq \tau\,.
\end{align*}

Thus, the proof of the lemma is complete. 
\end{proof}
}

\ignore{

\begin{lemma} \label{lem:Svals} Given our notation earlier, and the following lower bounds for the $r_\ell$'s:
\begin{align*}
	&r_\ell \geq \frac{8 \cdot 42^2 \, \left(\ln(42/\delta) + (\log b_L - \log b_{\ell-1})^{1/8}\right) \cdot (\log b_L - \log b_{\ell-1})^{5/2}}{\epsilon^2}
	\quad\quad\quad\quad \forall \ell \in [L-1]\,, \text{ and}
	\\&r_L \geq r_{L-1}\,.
\end{align*}

there exists a sequence of numbers $\{s_\ell\}_{\ell=1}^{L-1}$ such that the $s_\ell$'s are  in $[0,1]$, and the following holds:
\begin{enumerate}
	\item[$i)$] $\sum_{\ell=1}^L s_\ell <\frac{\epsilon}{3}\,,$
	\item[$ii)$] $
	\sum_{\ell=1}^{L} 3\exp\left(-\frac{r_\ell\,s_\ell^2 }{8\,(\log b_\ell - \log b_{\ell-1})^2}\right)< \delta\,.$
\end{enumerate}
\end{lemma}

\begin{proof} We prove this lemma with a similar proof to proof of Lemma~\ref{lem:Tvals}. We propose the following values for the sequence:
\begin{align*}
&s_\ell \coloneqq \frac{\epsilon}{42 (\log b_L - \log b_{\ell-1})^{1/4}}
	\quad\quad\quad\quad \forall \ell \in [L-1]\,, \text{ and}
	\\&s_L \coloneqq s_{L-1}
\,.
\end{align*}

Using Equation~\eqref{eq:boundDeltaLogBL}, it is clear that $s_\ell < 1$ for all $\ell$'s. 

We start by bounding the left-hand side of item $(i)$:
\begin{align*}
	\sum_{\ell=1}^{L} s_\ell & = s_{L} + \sum_{\ell=1}^{L-1} s_\ell  
	= s_{L-1} + \sum_{\ell=1}^{L-1} s_\ell 
	\leq 2 \sum_{\ell=1}^{L-1} s_\ell 
	 = \frac{\epsilon}{21} \sum_{\ell=1}^{L-1} \frac{1}{\left(\log b_L - \log b_{\ell-1})\right)^{1/4}} 
	\\ & \leq \frac{\epsilon}{21} \sum_{\ell=1}^{L-1} \frac{1}{\left(\log^{(\ell)}(k)\right)^{1/4}} <
	\frac{\epsilon}{3}
	\,.
\end{align*}
where for the last inequality, we use Proposition~\ref{prop:convergentSeries}.

Next, we prove item $(ii)$. Using the lower bound on the $r_\ell$'s and the fact that $b_1 < b_2 < \cdots < b_L$, we get: 
\begin{align*}\sum_{\ell=1}^{L} 3\exp\left(-\frac{r_\ell\,s_\ell^2 }{8\,(\log b_\ell - \log b_{\ell-1})^2}\right)
& = 3\exp\left(-\frac{r_L\,s_L^2 }{8\,(\log b_L - \log b_{L-1})^2}\right)
+
\sum_{\ell=1}^{L-1}  3\exp\left(-\frac{r_\ell\,s_\ell^2 }{8\,(\log b_\ell - \log b_{\ell-1})^2}\right)
\\ & \leq  3\exp\left(-\frac{r_L\,s_L^2 }{8\,(\log b_L - \log b_{L-1})^2}\right)
+
\sum_{\ell=1}^{L-1}  3\exp\left(-\frac{r_\ell\,s_\ell^2 }{8\,(\log b_L - \log b_{\ell-1})^2}\right)
\,.
\end{align*}
Note that in the last line the first term is equal to the last term in summation. Thus, we can bound the above by two times the summation:
\begin{align*}
	\sum_{\ell=1}^{L} 3\exp\left(-\frac{r_\ell\,s_\ell^2 }{8\,(\log b_\ell - \log b_{\ell-1})^2}\right)
	& \leq  6
\sum_{\ell=1}^{L-1}  \exp\left(-\frac{r_\ell\,s_\ell^2 }{8\,(\log b_L - \log b_{\ell-1})^2}\right)
\\ & \leq 6 \sum_{\ell = 1}^{L-1}  \exp \left(-\log(42/\delta) - (\log b_L - \log b_{\ell-1})^{1/8}
		\right)
	\\ & \leq
	\frac{\delta}{7}  \sum_{\ell=1}^{L-1} \frac{1}{\left(\log b_L - \log b_{\ell-1})\right)^{1/4}}
	< \delta\,.
\end{align*}
where for the second to last inequality, we use $e^{-x} \leq x^{-2}$, and  
the last inequality is obtained using Proposition~\ref{prop:convergentSeries}. Hence, the proof is complete. 
\end{proof}

\begin{lemma} \label{lem:Tvals} Given our notation earlier, and the following lower bounds for the $r_\ell$'s:
\begin{align*}
	r_\ell \geq \frac{21^2 \, \left(\ln(14/\delta) + (\log b_L - \log b_{\ell-1})^{1/8}\right) \cdot (\log b_L - \log b_{\ell-1})^{5/2}}{2\,\epsilon^2} \quad\quad\quad\quad \forall \ell \in [L-1]\,,
\end{align*}
there exists a sequence of numbers $\{t_\ell\}_{\ell=1}^{L-1}$ such that the $t_\ell$'s are  in $[0,1]$, and the following holds:
\begin{itemize}
	\item[$(i)$] $\sum_{\ell=1}^{L-1} t_\ell \cdot (\log b_L - \log b_{\ell-1}) < \epsilon/3\,,$
	\item[$(ii)$] $
	\sum_{\ell=1}^{L-1} 2\exp\left(-2\,r_\ell \, t_\ell^2\right) < \delta\,.$
\end{itemize}
\end{lemma}

\begin{proof} We propose the following values for the sequence. Set $t_\ell$ as follows:
	$$t_\ell \coloneqq \frac{\epsilon}{21 (\log b_L - \log b_{\ell-1})^{5/4}}\,.$$

Observe the following bound, which we use in this proof:
\begin{equation}\label{eq:boundDeltaLogBL}
	\log b_L - \log b_{\ell-1} = \log \left(b_L/b_{\ell-1}\right) = \log \left(\frac{6(\log^{(\ell-1)}(k))^3}{\epsilon}\right) = 3 \, \log^{(\ell)}(k) + \log (6/\epsilon) \geq \max(1, \log^{(\ell)}(k))
	\,.
\end{equation}
Using the above bound, it is clear that $t_\ell < 1$ for all $\ell$'s. 

We start by bounding the left-hand side of item $(i)$:
\begin{align*}
	\sum_{\ell=1}^{L-1} t_\ell \cdot (\log b_L - \log b_{\ell-1}) 
	& = \frac{\epsilon}{21} \sum_{\ell=1}^{L-1} \frac{1}{\left(\log b_L - \log b_{\ell-1})\right)^{1/4}} 
	\leq \frac{\epsilon}{21} \sum_{\ell=1}^{L-1} \frac{1}{\left(\log^{(\ell)}(k)\right)^{1/4}} < \frac{\epsilon}{3}
	\,.
\end{align*}
where for the last inequality, we use Proposition~\ref{prop:convergentSeries}.

Next, we prove item $(ii)$ using the lower bound on $r_\ell$ and the fact that $e^{-x} \leq x^{-2}$:
\begin{align*}
	\sum_{\ell=1}^{L-1} 2\exp\left(-2\,r_\ell \, t_\ell^2\right) 
	& \leq \sum_{\ell = 1}^{L-1} 2 \exp \left(-\log(14/\delta) - (\log b_L - \log b_{\ell-1})^{1/8}
		\right)
	\\ & \leq
	\frac{\delta}{7}  \sum_{\ell=1}^{L-1} \frac{1}{\left(\log b_L - \log b_{\ell-1})\right)^{1/4}}
	< \delta\,,
\end{align*}
where the last inequality is obtained using the same calculation as item ($i$). Hence, the proof is complete. 
\end{proof}

\begin{proposition} \label{prop:convergentSeries} For a positive integer $k$, let $L = \log^* (k)$. Then, we have:
$\sum_{\ell=1}^{L-1} \left(\log^{(\ell)}(k)\right)^{-1/4}
	< 7$.
\end{proposition}
\begin{proof}
	We claim $\log^{(L-i)}(k) > 2^{i-1}$ for $i \in \{1,2,\ldots,L-1\}$. One can show this inequality by induction. For $i=1$, we use the definition of $L = \log^*(k)$. Note that $\log^*(k)$ is the number of times we apply the logarithm function before we reach a number that is at most one. Thus, 
$\log^{(L-1)}(k)$ is greater than one. For $i > 1$, using the induction hypothesis for $i-1$, we obtain:
\begin{align*}
	\log^{(L-i)}(k) =  2^{\log^{(L-(i-1))}(k)} > 2^{2^{i-2}} \geq 2^{i-1}
	\,.
\end{align*}
Now, we focus on the series in the statement of the lemma:
\begin{align*}
	\sum_{\ell=1}^{L-1} \frac{1}{\left(\log^{(\ell)}(k)\right)^{1/4}}
	 \leq  \sum_{\ell=1}^{L-1} 2^{-(\ell-1)/4}
	\leq \sum_{\ell=1}^{\infty} 2^{-(\ell-1)/4} =  \frac{2}{2 - 2^{3/4}} < 7
	\,.
\end{align*}
Thus, the proof is complete.
\end{proof}
}


\begin{thebibliography}{JVHW15}

\bibitem[ABIS19]{AcharyaBIS19}
Jayadev Acharya, Sourbh Bhadane, Piotr Indyk, and Ziteng Sun.
\newblock Estimating entropy of distributions in constant space.
\newblock In {\em Proceedings of the 32nd Annual Conference on Neural
  Information Processing (NeurIPS)}, pages 5163--5174, 2019.

\bibitem[AMS99]{AlonMS99}
Noga Alon, Yossi Matias, and Mario Szegedy.
\newblock The space complexity of approximating the frequency moments.
\newblock {\em J. Comput. Syst. Sci.}, 58(1):137--147, 1999.

\bibitem[BDKR05]{BatuDKR05}
Tugkan Batu, Sanjoy Dasgupta, Ravi Kumar, and Ronitt Rubinfeld.
\newblock The complexity of approximating the entropy.
\newblock {\em {SIAM} J. Comput.}, 35(1):132--150, 2005.

\bibitem[BG06]{BhuvanagiriG06}
Lakshminath Bhuvanagiri and Sumit Ganguly.
\newblock Estimating entropy over data streams.
\newblock In {\em Algorithms - {ESA} 2006, 14th Annual European Symposium,
  Zurich, Switzerland, September 11-13, 2006, Proceedings}, pages 148--159,
  2006.

\bibitem[CCM10]{ChakrabartiCM10}
Amit Chakrabarti, Graham Cormode, and Andrew McGregor.
\newblock A near-optimal algorithm for estimating the entropy of a stream.
\newblock {\em {ACM} Trans. Algorithms}, 6(3):51:1--51:21, 2010.

\bibitem[CLM10]{ChienLM10}
Steve Chien, Katrina Ligett, and Andrew McGregor.
\newblock Space-efficient estimation of robust statistics and distribution
  testing.
\newblock In {\em Proceedings of the 1st Annual Conference Innovations in
  Computer Science (ICS)}, pages 251--265, 2010.

\bibitem[GKLR21]{GargKLR21}
Sumegha Garg, Pravesh~K. Kothari, Pengda Liu, and Ran Raz.
\newblock Memory-sample lower bounds for learning parity with noise.
\newblock In {\em Proceedings of the 25th International Workshop on
  Randomization and Computation (RANDOM)}, pages 60:1--60:19, 2021.

\bibitem[GKR20]{GargKR20}
Sumegha Garg, Pravesh~K. Kothari, and Ran Raz.
\newblock Time-space tradeoffs for distinguishing distributions and
  applications to security of goldreich's {PRG}.
\newblock In {\em Proceedings of the 24th International Workshop on
  Randomization and Computation (RANDOM)}, pages 21:1--21:18, 2020.

\bibitem[GM07]{GuhaM07}
Sudipto Guha and Andrew McGregor.
\newblock Space-efficient sampling.
\newblock In {\em Proceedings of the 11th International Conference on
  Artificial Intelligence and Statistics (AISTATS)}, pages 171--178, 2007.

\bibitem[GRT18]{GargRT18}
Sumegha Garg, Ran Raz, and Avishay Tal.
\newblock Extractor-based time-space lower bounds for learning.
\newblock In {\em Proceedings of the 50th Annual {ACM} {SIGACT} Symposium on
  Theory of Computing (STOC)}, pages 990--1002, 2018.

\bibitem[GRT19]{GargRT19}
Sumegha Garg, Ran Raz, and Avishay Tal.
\newblock Time-space lower bounds for two-pass learning.
\newblock In {\em Proceedings of the 34th Computational Complexity Conference
  (CCC)}, pages 22:1--22:39, 2019.

\bibitem[HNO08]{HarveyNO08}
Nicholas J.~A. Harvey, Jelani Nelson, and Krzysztof Onak.
\newblock Sketching and streaming entropy via approximation theory.
\newblock In {\em Proceedings of the 49th Annual {IEEE} Symposium on
  Foundations of Computer Science (IEEE)}, pages 489--498, 2008.

\bibitem[JVHW15]{JiaoVHW15}
Jiantao Jiao, Kartik Venkat, Yanjun Han, and Tsachy Weissman.
\newblock Minimax estimation of functionals of discrete distributions.
\newblock {\em IEEE Trans. Inf. Theory}, 61(5):2835--2885, 2015.

\bibitem[KRT17]{KolRT17}
Gillat Kol, Ran Raz, and Avishay Tal.
\newblock Time-space hardness of learning sparse parities.
\newblock In {\em Proceedings of the 49th Annual {ACM} {SIGACT} Symposium on
  Theory of Computing (STOC)}, pages 1067--1080, 2017.

\bibitem[MG82]{MisraG82}
Jayadev Misra and David Gries.
\newblock Finding repeated elements.
\newblock {\em Sci. Comput. Program.}, 2(2):143--152, 1982.

\bibitem[MM17]{MoshkovitzM17}
Dana Moshkovitz and Michal Moshkovitz.
\newblock Mixing implies lower bounds for space bounded learning.
\newblock In {\em Proceedings of the 30th Conference on Learning Theory
  (COLT)}, pages 1516--1566, 2017.

\bibitem[MP80]{MunroP80}
J.~Ian Munro and Mike Paterson.
\newblock Selection and sorting with limited storage.
\newblock {\em Theor. Comput. Sci.}, 12:315--323, 1980.

\bibitem[Raz17]{Raz17}
Ran Raz.
\newblock A time-space lower bound for a large class of learning problems.
\newblock In {\em Proceedings of the 58th {IEEE} Annual Symposium on
  Foundations of Computer Science (FOCS)}, pages 732--742, 2017.

\bibitem[Raz19]{Raz19}
Ran Raz.
\newblock Fast learning requires good memory: {A} time-space lower bound for
  parity learning.
\newblock {\em J. {ACM}}, 66(1):3:1--3:18, 2019.

\bibitem[She13]{Sherstov13}
Alexander~A. Sherstov.
\newblock Making polynomials robust to noise.
\newblock {\em Theory Comput.}, 9:593--615, 2013.

\bibitem[SSV19]{SharanSV19}
Vatsal Sharan, Aaron Sidford, and Gregory Valiant.
\newblock Memory-sample tradeoffs for linear regression with small error.
\newblock In {\em Proceedings of the 51st Annual {ACM} {SIGACT} Symposium on
  Theory of Computing (STOC)}, pages 890--901, 2019.

\bibitem[SVW16]{SteinhardtVW16}
Jacob Steinhardt, Gregory Valiant, and Stefan Wager.
\newblock Memory, communication, and statistical queries.
\newblock In {\em Proceedings of the 29th Conference on Learning Theory
  (COLT)}, pages 1490--1516, 2016.

\bibitem[VV11]{ValiantV11}
Gregory Valiant and Paul Valiant.
\newblock The power of linear estimators.
\newblock In {\em Proceedings of the 52nd Annual {IEEE} Symposium on
  Foundations of Computer Science (FOCS)}, pages 403--412, 2011.

\bibitem[VV17]{ValiantV17}
Gregory Valiant and Paul Valiant.
\newblock Estimating the unseen: Improved estimators for entropy and other
  properties.
\newblock {\em J. {ACM}}, 64(6):37:1--37:41, 2017.

\bibitem[Wei]{wikipolylogarithm}
Eric~W. Weisstein.
\newblock Polylogarithm. from {M}ath{W}orld--{A} {W}olfram {W}eb resource.
\newblock Last accessed Feb. 9 2022.

\bibitem[WY16]{WuY16}
Yihong Wu and Pengkun Yang.
\newblock Minimax rates of entropy estimation on large alphabets via best
  polynomial approximation.
\newblock {\em IEEE Trans. Inf. Theory}, 62(6):3702--3720, 2016.

\end{thebibliography}
\end{document}